\documentclass[
final
 , nomarks
]{dmtcs-episciences}

\usepackage{amsthm,amsmath,cite,float}
\usepackage{graphicx}
\usepackage[shortlabels]{enumitem}
\usepackage[table]{xcolor}
\usepackage[utf8]{inputenc}
\usepackage{subfigure}
\newtheorem {theorem}{Theorem}
\newtheorem {lemma}{Lemma}
\newtheorem {proposition}{Proposition}
\newtheorem {corollary}{Corollary}

\newcommand{\commentout}[1]{}
\renewcommand{\thefootnote}{\arabic{footnote}}
\usepackage[numbers]{natbib}
\author{Feodor F. Dragan
  \and Abdulhakeem Mohammed
  }
\title{Slimness of graphs}
\affiliation{
Algorithmic Research Laboratory, Department of Computer Science,  Kent State University, Kent, OH, USA}

\received{2018-2-15}
\accepted{2019-2-20}

\begin{document}
\publicationdetails{21}{2019}{3}{10}{4288}
\maketitle
\begin{abstract}
\emph{Slimness} of a graph measures the local deviation of its metric from a tree metric. In a graph $G=(V,E)$, a \emph{ geodesic triangle} $\bigtriangleup(x,y,z)$ with $x, y, z\in V$ is the union $P(x,y) \cup P(x,z) \cup P(y,z)$ of three shortest paths connecting these vertices. A geodesic triangle $\bigtriangleup(x,y,z)$ is called $\delta$-\emph{slim} if for any vertex $u\in V$ on any side $P(x,y)$ the distance from $u$ to $P(x,z) \cup P(y,z)$ is at most $\delta$, i.e. each path is contained in the union of the $\delta$-neighborhoods of two others. A graph  $G$ is called $\delta$-\emph{slim}, if  all geodesic triangles in $G$ are $\delta$-\emph{slim}. The smallest value $\delta$ for which $G$  is $\delta$-\emph{slim} is called \emph{ the slimness} of $G$.
In this paper, using the layering partition technique, we obtain sharp bounds on slimness of such families of graphs as (1) graphs with cluster-diameter $\Delta(G)$ of a layering partition of $G$, (2) graphs with tree-length $\lambda$, (3) graphs with tree-breadth $\rho$,  (4) $k$-chordal graphs, AT-free graphs and HHD-free graphs. Additionally, we show  that the slimness of every 4-chordal graph is at most 2 and characterize those 4-chordal graphs for which the slimness of every of its induced subgraph  is at most 1.
\end{abstract}
\section{Introduction}
Recently, there has been a surge of empirical works measuring and analyzing geometric characteristics of real-world networks, namely the \emph {hyperbo\-licity} (sometimes called also the \emph{global negative curvature}) of the network (see, e.g., \cite{DBLP:journals/networks/Abu-AtaD16,DBLP:conf/icdm/AdcockSM13,conf/isaac/ChenFHM12,Kennedy2013Arch,conf/nca/MontgolfierSV11,narayan2011large,Jonckheere:2008:SGH:1330411.1330417,DBLP:journals/im/JonckheereLBB11,DBLP:journals/ton/ShavittT08}). Hyperbolicity measures the local deviation of a metric from a tree metric. It has been shown that a number of data networks, including Internet application networks, web networks, collaboration networks, social networks, and others, have small hyperbolicity.
It has been confirmed (see \cite{DBLP:conf/soda/ChepoiDV17}) that the property, observed in real-world networks (see \cite{DBLP:journals/im/JonckheereLBB11,narayan2011large}),  in which traffic between vertices (nodes) tends to go through a
relatively small core of the network, as if the shortest path between them is curved inwards, is indeed due to global negative curvature of the network.

Fortunately, graphs and general geodesic metric spaces with small hyperbolicities have many algorithmic advantages. They allow efficient approximate solutions for a number of optimization problems. For example, Krauthgamer and Lee~\cite{KrLe} presented a PTAS for the Traveling Salesman Problem when the set of cities lie in a hyperbolic metric space. Chepoi and Estellon~\cite{DBLP:conf/approx/ChepoiE07} established a relationship between the minimum number of balls of
radius $r+2\delta$ covering a finite subset $S$ of a
$\delta$-hyperbolic geodesic space and the size of the maximum
$r$-packing of $S$ and showed how to compute such coverings and
packings in polynomial time. Edwards et al. \cite{Edwards2018FastAA} provided a quasilinear time algorithm for the $p$-center problem with an additive
error less than or equal to 3 times the input graph's hyperbolicity. Chepoi et al.~\cite{DBLP:conf/compgeom/ChepoiDEHV08}
gave efficient algorithms for quick and accurate estimations of diameters
and radii of $\delta$-hyperbolic geodesic spaces and graphs. Additionally, Chepoi et al.~\cite{ChDrEsRout} showed that every $n$-vertex $\delta$-hyperbolic
graph has an additive $O(\delta \log n)$-spanner with at most
$O(\delta n)$ edges and enjoys an $O(\delta\log
n)$-additive routing labeling scheme with $O(\delta\log^2n)$-bit
labels and $O(\log\delta)$ time routing protocol.
Efficient embeddings of hyperbolic graphs into hyperbolic spaces are provided by Verbeek and Suri \cite{DBLP:conf/compgeom/VerbeekS14}. Effect of hyperbolicity parameter on cuts and expansions in graphs with some algorithmic implications is considered by DasGupta et al. \cite{DBLP:journals/corr/DasGuptaKMY15}.

In case of graphs (and general geodesic metric spaces), there exist several "equivalent" definitions of $\delta$-hyperbolicity involving different but
comparable values of $\delta$ \cite{alonso1991notes,GhHa,Gromov87}. In this paper, we are interested in two of them, in Gromov's  4-point condition and in Rips'  condition involving geodesic triangles. Let $G=(V,E)$ be a graph, $d(\cdot,\cdot)$ be the shortest path metric defined on $V$, and $\delta\ge 0$.
\commentout{The \emph {Gromov product} of $y,z\in V$ with respect to $w$ is defined to be
	$$(y|z)_w=\frac{1}{2}(d(y,w)+d(z,w)-d(y,z)).$$
	Let $\delta\ge 0$. A $G=(V,E)$  is said to be $\delta$-\emph { hyperbolic} \cite{Gromov87} if
	$$(x|y)_w\ge \min \{ (x|z)_w, (y|z)_w\}-\delta$$
	for all $w,x,y,z\in V$. Equivalently, $G=(V,E)$  is $\delta$-hyperbolic
	if  for any four vertices $u,v,x,y$ of $V$, the two larger of the three distance sums
	$d(u,v)+d(x,y)$, $d(u,x)+d(v,y)$, $d(u,y)+d(v,x)$ differ by at most
	$2\delta \geq 0$. } 

A graph $G=(V,E)$  is said to be $\delta$-\emph{hyperbolic}\cite{Gromov87} if
for any four vertices $u,v,x,y$ of $V$, the two larger of the three distance sums
$d(u,v)+d(x,y)$, $d(u,x)+d(v,y)$, $d(u,y)+d(v,x)$ differ by at most
$2\delta$. The smallest value $\delta$ for which $G$  is $\delta$-hyperbolic is called \emph{the hyperbolicity} of $G$ and denoted by $hb(G)$.

A \emph{geodesic triangle} $\bigtriangleup(x,y,z)$ with $x, y, z\in V$ is the union $P(x,y) \cup P(x,z) \cup P(y,z)$ of three shortest paths connecting these vertices. A geodesic triangle $\bigtriangleup(x,y,z)$ is called $\delta$-\emph{slim} if for any vertex $u\in V$ on any side $P(x,y)$ the distance from $u$ to $P(x,z) \cup P(y,z)$ is at most $\delta$, i.e. 
each path is contained in the union of the $\delta$-neighborhoods of two others (see Figure \ref{fig:slimContained}). We say that a graph $G$ is  $\delta$-\emph{slim}, if  all geodesic triangles in $G$ are $\delta$-slim. The smallest value $\delta$ for which $G$  is $\delta$-slim is called \emph{ the slimness} of $G$ and denoted by $sl(G)$.

\begin{figure}[htbp]
	\centering
	\includegraphics[scale=0.21]{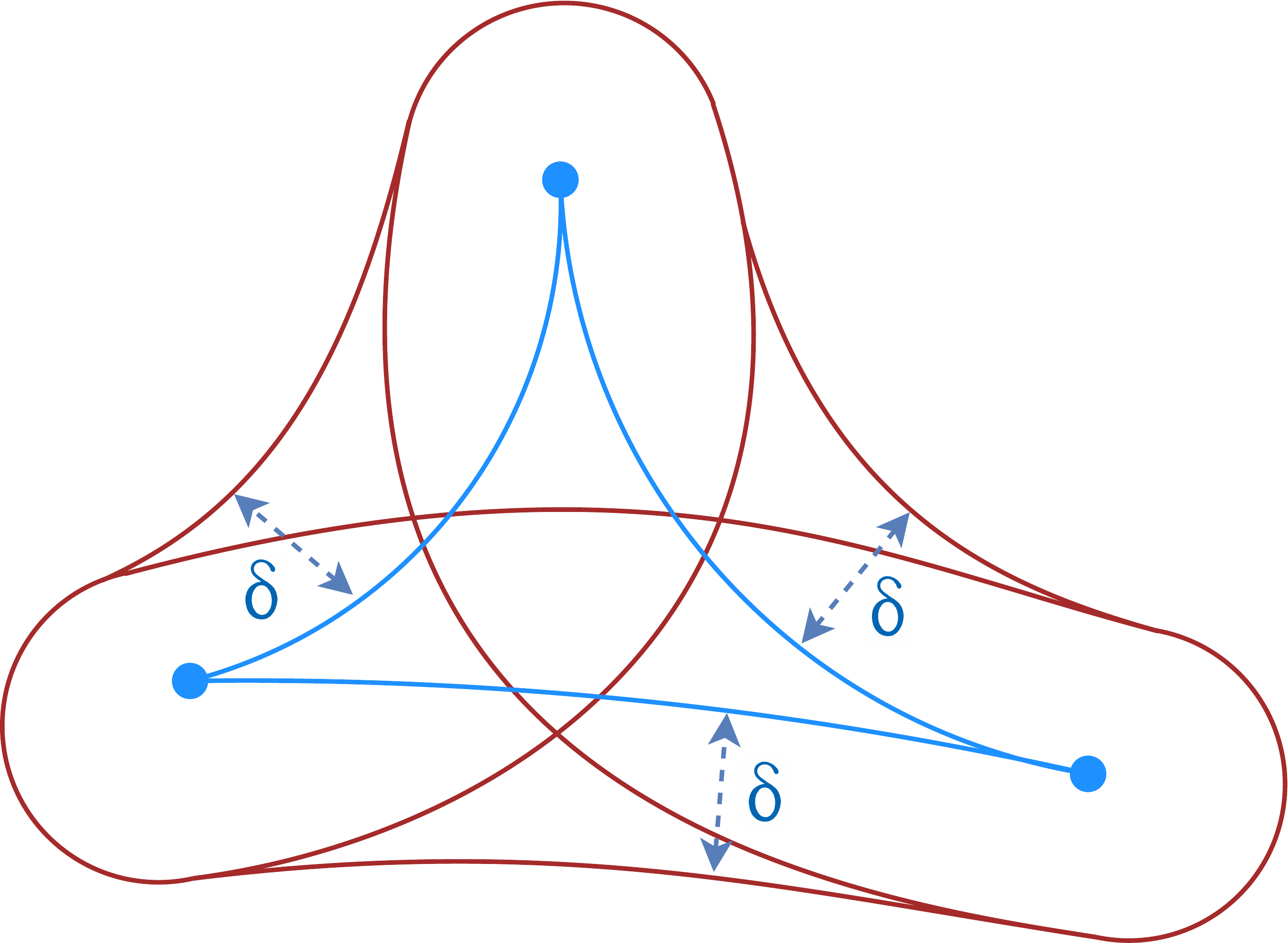}
	\caption{ A $\delta$-\emph{slim} triangle.}
	\label{fig:slimContained}
\end{figure}

It is known that for every graph $G$,  $hb(G) \leq 2sl(G)+ \frac{1}{2}$ \cite{Soto} and $sl(G) \leq 3hb(G)+ \frac{1}{2}$ \cite{alonso1991notes} and these inequalities are sharp for general graphs. It is clear from the definition that the hyperbolicity $hb(G)$ of an $n$-vertex graph $G$ can be computed in at most $O(n^4)$ time. It is less obvious, however, that the slimness $sl(G)$ of $G$ can also be computed in at most $O(n^4)$ time~\cite{chalopin2018fast_SoCG2018}.


Motivated by the abundance of real-world networks with small hyperbolicities and by many algorithmic advantages that small hyperbolicity provides, a few researchers started investigating possible bounds on the hyperbolicity $hb(G)$ of special classes of graphs (see \cite{DBLP:journals/networks/Abu-AtaD16,DBLP:journals/siamdm/BandeltC03,DBLP:journals/siamdm/ChalopinCPP14,DBLP:conf/compgeom/ChepoiDEHV08,koolen2002hyperbolic,uea21813,DBLP:journals/combinatorics/WuZ11}).  ($\frac{1}{2}$)-Hyperbolic graphs were characterized in \cite{DBLP:journals/siamdm/BandeltC03} and then later another characterization with more direct algorithmic applications was proposed in \cite{coudert2014recognition}. 
The hyperbolicity of chordal graphs was investigated in \cite{uea21813}. It was shown that the hyperbolicity of a chordal graph is at most 1. Furthermore,  chordal graphs with hyperbolicity at most $\frac{1}{2}$ were characterized by two forbidden isometric subgraphs \cite{uea21813}.  
Later, $k$-chordal graphs were studied in \cite{DBLP:journals/combinatorics/WuZ11}. It was proven that every $k$-chordal graph ($k\ge 4$) has hyperbolicity at most $\frac{\lfloor\frac{k}{2}\rfloor}{2}$ \cite{DBLP:journals/combinatorics/WuZ11}. Additionally, all 4-chordal (5-chordal) graphs with hyperbolicity at most $\frac{1}{2}$ were characterized by three (five, respectively) forbidden isometric subgraphs \cite{DBLP:journals/combinatorics/WuZ11}. The hyperbolicity of graphs with tree-length $\lambda$ was considered in \cite{DBLP:conf/compgeom/ChepoiDEHV08}. It was shown that it is at most $\lambda$.  Recently, in \cite{DBLP:journals/networks/Abu-AtaD16}, also graphs with bounded cluster-diameter $\Delta(G)$ (the minimum cluster-diameter over all layering partitions of $G$; see Section \ref{sec:relations} for more details) were considered. It was shown that their hyperbolicity is at most $\Delta(G)$.

\begin{table} [htbp]
	\centering
	\begin{tabular}{|l|l|l|}
		\hline
		\rowcolor{lightgray} Graph Class                & Hyperbolicity & Slimness \\ \hline\noalign{\smallskip}
		graphs with cluster-diameter $\Delta(G)$ & $\leq\Delta(G)$\cite{DBLP:journals/networks/Abu-AtaD16} & $\leq\lfloor\frac{3}{2}\Delta(G)\rfloor$ [here] \\ \hline\noalign{\smallskip}	
		graphs with tree-length $\lambda$ & $\leq\lambda$\cite{DBLP:conf/compgeom/ChepoiDEHV08} & $\leq\lfloor\frac{3}{2}\lambda\rfloor$ \cite{diestel2012connectedTW}, [here] \\ \hline\noalign{\smallskip}	
		graphs with tree-breadth $\rho$ & $\leq2\rho$ (follows from \cite{DBLP:conf/compgeom/ChepoiDEHV08}) & $\leq 3 \rho$ [here] \\ \hline			 k-chordal graphs ($k\ge 4$) & $\leq\frac{\lfloor\frac{k}{2}\rfloor}{2}$\cite{DBLP:journals/combinatorics/WuZ11} &$\leq\lfloor\frac{k}{4}\rfloor+1$ \cite{bermudo2016hyperbolicity}, [here] \\ \hline
		4-chordal graphs     & $\leq1$ \cite{DBLP:journals/combinatorics/WuZ11}              & $\leq2$ [here]         \\ \hline
		HHD-free  graphs          & $\leq1$  \cite{DBLP:journals/combinatorics/WuZ11}             & $\leq1$  [here]        \\  \hline
		AT-free graphs     & $\leq1$   \cite{DBLP:journals/combinatorics/WuZ11}          & $\leq1$ [here]         \\ \hline
		Chordal graphs            & $\leq1$ \cite{uea21813}            & $\leq1$  [here]        \\  \hline
		Block graphs              & 0 [folklore]            & 0  [folklore]       \\  \hline
		
	\end{tabular}
	\caption{Hyperbolicity and slimness of some structured graph classes.}
	\label{table:hb_N_sl_G}
\end{table}

Much less is known about the slimness of particular classes of graphs. One can get straightforward bounds using the general inequality $sl(G) \leq 3hb(G)+ \frac{1}{2}$ and known bounds on $hb(G)$, but the bounds on $sl(G)$ obtained this way are usually far from being sharp.
In \cite{diestel2012connectedTW}, a non-trivial result was obtained by Diestel and M\"uller. It was shown that the slimness of graphs with tree-length $\lambda$ is at most $\lfloor\frac{3}{2}\lambda\rfloor$, and the result is sharp for every $\lambda\ge 1$. Recently, Bermudo et al. \cite{bermudo2016hyperbolicity} proved that the slimness of a $k$-chordal graph is at most $\frac{k}{4}+1$.

In this paper, using the layering partition technique, in an unified way we obtain sharp bounds on slimness of such families of graphs as (1) graphs with cluster-diameter $\Delta(G)$ of a layering partition of $G$, (2) graphs with tree-length $\lambda$, (3) graphs with tree-breadth $\rho$,  (4) $k$-chordal graphs, AT-free graphs and HHD-free graphs. Additionally, we show  that the slimness of every 4-chordal graph is at most 2 and characterize those 4-chordal graphs for which the slimness of every of its induced subgraph  is at most 1. Table \ref{table:hb_N_sl_G} summarizes our results for slimness and known results for hyperbolicity in special graph classes.

\subsection{Basic notions and notations}
All graphs appearing here are connected, finite, unweighted, undirected, loopless and without multiple edges. 
For a graph $G=(V,E)$, we use $n$ and $|V|$ interchangeably to denote the number of vertices in $G$. Also, we use $m$ and $|E|$ to denote the number of edges. A path $P$ of length $k$ in a graph $G$ is a sequence of vertices $(v_0,v_1,\dots,v_k)$ such that $v_i$ is adjacent to $v_{i+1}$ for each $i, 0\leq i< k$. 
The {\em distance} $d_G(u,v)$ between vertices $u$ and $v$ is the length of a shortest path connecting $u$ and $v$ in $G$. The distance $d_G(v,M)$ between a vertex $v$ and a set $M\subseteq V$ is defined by $d_G(v,M)=\min\{d_G(v,u): u\in M\}$.   The {\em ball} $B_r(s,G)$ of a graph $G$ centered at a vertex $s \in V$ and with radius $r$ is the set of all vertices with distance no more than $r$ from $s$ (i.e., $B_r(s,G)=\{v\in V: d_G(v,s) \leq r \}$). We omit the graph name $G$ and write  $B_r(s)$ if the context is about only one graph. For any two vertices $u$, $v$ of $G$, $I(u,v)= \lbrace z\in V:d(u,v)=d(u,z)+d(z,v) \rbrace$ is the (metric) {\sl interval} between $u$ and $v$, i.e., all vertices that lay on shortest paths between $u$ and $v$.
A subgraph $H$ of a graph $G$ is called {\em isometric} if for every two vertices $u,v$ of $H$, $d_G(u,v)=d_H(u,v)$ holds.
%
%
The \emph{eccentricity} of a vertex $v$, denoted by $ecc_G(v)$, is the largest distance from that vertex $v$ to any other vertex, i.e., $ecc_G(v)=\max_{u \in V} d_G(v,u)$.  The \emph{radius} $rad(G)$ of a graph $G$ is the minimum eccentricity of a vertex in $G$, i.e., $rad(G)=\min_{v \in V} \max_{u\in V}d_G(v,u)$.

Definitions of other graph parameters considered, as well as notions and notation local to a section, are given in appropriate sections.

\section{Slimness and other tree-likeness parameters}	\label{sec:relations}
In this section, we establish a relation between the slimness and the cluster-diameter of a layering partition of a graph. As a corollary, we get a recent result of Diestel and M\"uller \cite{diestel2012connectedTW} on a relation between the slimness and the tree-length of a graph.

\subsection{Slimness and cluster-diameter of a layering partition} \label{subsec:cl-diam}
Layering partition is a graph decomposition procedure introduced in~\cite{DBLP:journals/jal/BrandstadtCD99,DBLP:journals/ejc/ChepoiD00} and used in~\cite{BaInSi,DBLP:journals/jal/BrandstadtCD99,DBLP:journals/ejc/ChepoiD00,ChepoiDNRV12} 
for embedding graph metrics into trees.

A \emph{layering} of a graph $G=(V, E)$ with respect to a start vertex $s$ is the decomposition of $V$ into $r+1$ layers (spheres), where $r=ecc_G(s)$ \begin{displaymath}L^i(s)=\{u\in V:d_G(s,u)=i\},i=0,1,\dots,r.\end{displaymath} A \emph{layering partition} \begin{displaymath}\mathcal{LP}(G,s)=\{L^i_1,\ldots,L^i_{p_i}:i=0,1,\dots,r\}.\end{displaymath} of $G$ is a partition of each layer $L^i(s)$ into clusters $L^i_1,\dots,L^i_{p_i}$ such that two vertices $u,v \in L^i(s)$ belong to the same cluster $L^i_j$ if and only if they can be connected by a path outside the ball $B_{i-1}(s)$ of radius $i-1$ centered at $s$. Here, $p_i$ is the number of clusters in layer $i$. See Figure \ref{fig:layering-partition} for an illustration. A layering partition of a graph can be constructed in $O(n+m)$ time (see~\cite{DBLP:journals/ejc/ChepoiD00}).

A \emph{layering tree} $\Gamma(G,s)$ of a graph $G$ with respect to a layering partition $\mathcal{LP}(G,s)$  is the graph whose nodes are the clusters of $\mathcal{LP}(G,s)$ and where two nodes $C=L_j^i$ and $C'=L_{j'}^{i'}$ are adjacent in $\Gamma(G,s)$ if and only if there exist a vertex $u \in C$ and a vertex $v\in C'$ such that $uv \in E$. It was shown in~\cite{DBLP:journals/jal/BrandstadtCD99} that the graph $\Gamma(G,s)$ is always a tree and, given a start vertex $s$,  it can be constructed in $O(n+m)$ time~\cite{DBLP:journals/ejc/ChepoiD00}. Note that, for a fixed start vertex $s\in V$, the layering partition $\mathcal{LP}(G,s)$ of $G$ and its tree $\Gamma(G,s)$ are unique.

The \emph{cluster-diameter $\Delta_s(G)$ of the layering partition $\mathcal{LP}(G,s)$ with respect to vertex $s$} is the largest diameter of a cluster in $\mathcal{LP}(G,s)$, i.e., \begin{displaymath}\Delta_s(G)=\max_{C \in \mathcal{LP}(G,s)} \max_{u,v\in C}d_G(u,v).\end{displaymath} 

The \emph{cluster-diameter $\Delta(G)$ of a graph $G$} is the minimum cluster-diameter over all layering partitions of $G$, i.e., \begin{displaymath}\Delta(G)=\min_{s \in V}\Delta_s(G).\end{displaymath} Let also  $\widehat{\Delta}(G)$ denote the maximum cluster-diameter over all layering partitions of $G$, i.e., \begin{displaymath}\widehat{\Delta}(G)=\max_{s \in V}\Delta_s(G).\end{displaymath}
	
	\begin{figure}[H]\footnotesize
		\centering                                                        
		\subfigure[][Layering of graph $G$ with respect to $s$.]
		{                    
			\scalebox{0.20}[0.20]{\includegraphics{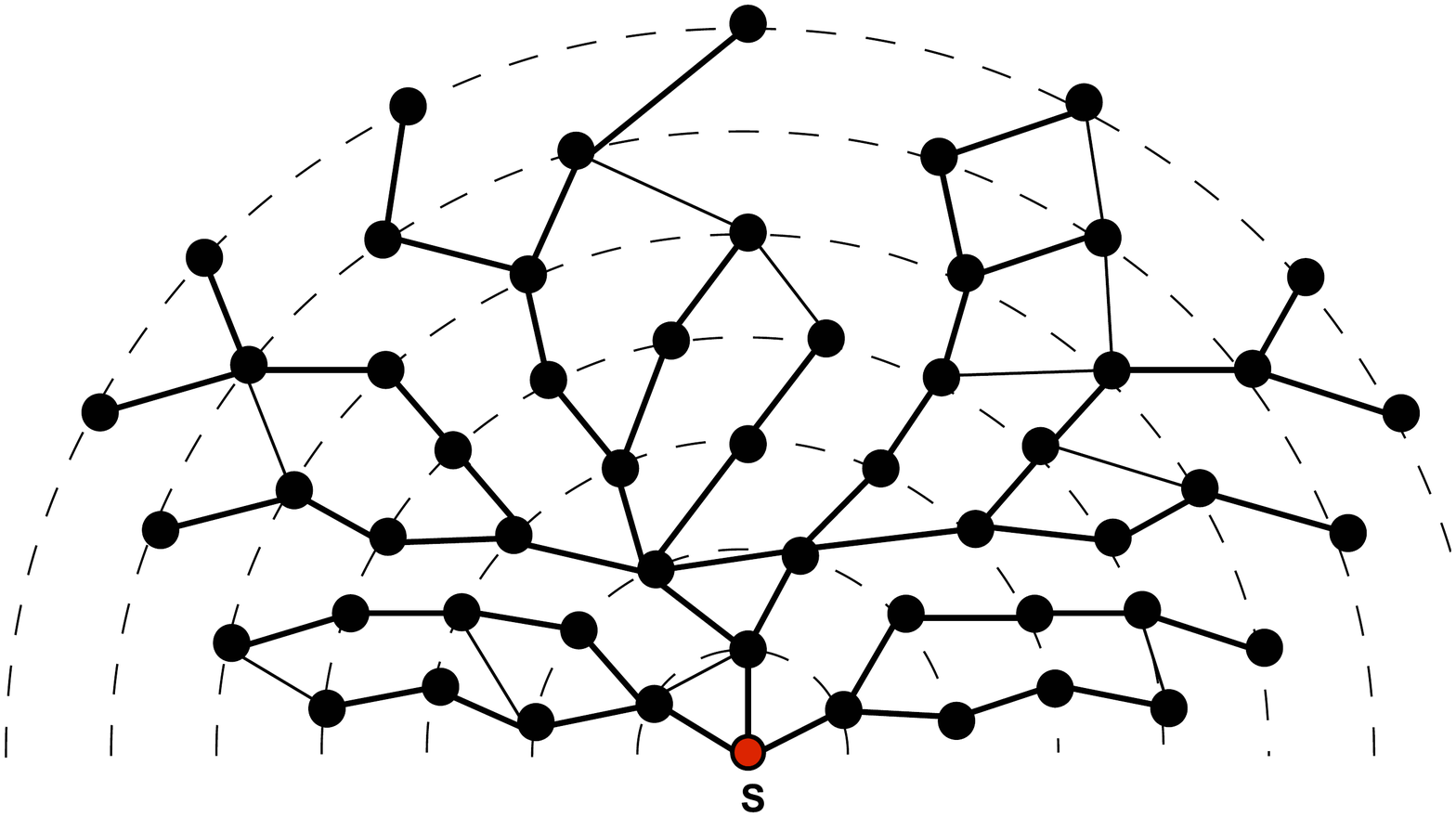}}
			\label{fig:layering}      
		}
		\subfigure[][ Clusters of the layering partition $\mathcal{LP}(G,s)$.]
		{
			\scalebox{0.20}[0.20]{\includegraphics{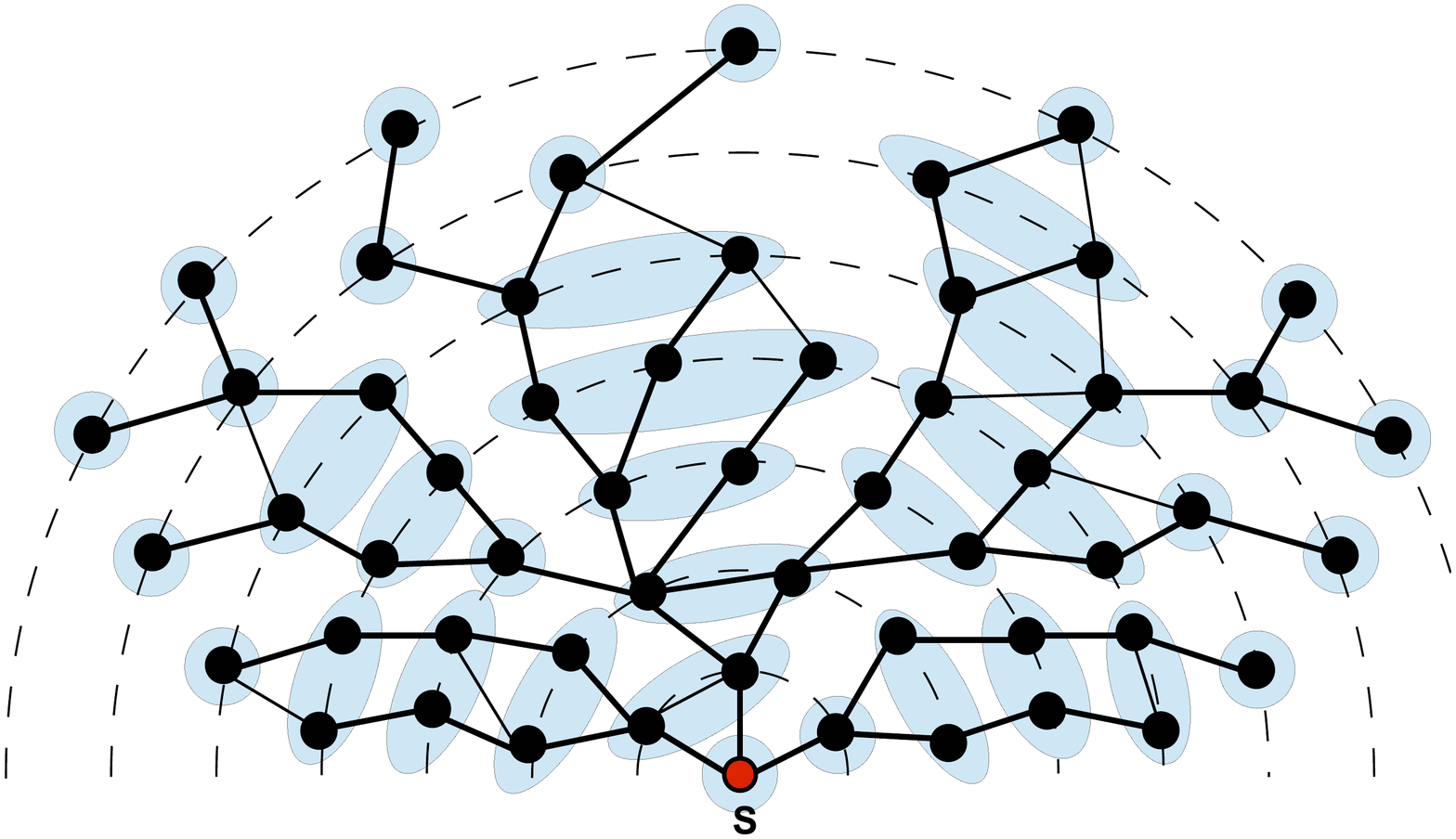}}
			\label{fig:Layering-clusters}
		}
		
		\subfigure[][Layering tree $\Gamma(G,s)$.]
		{                    
			\scalebox{0.20}[0.20]{\includegraphics{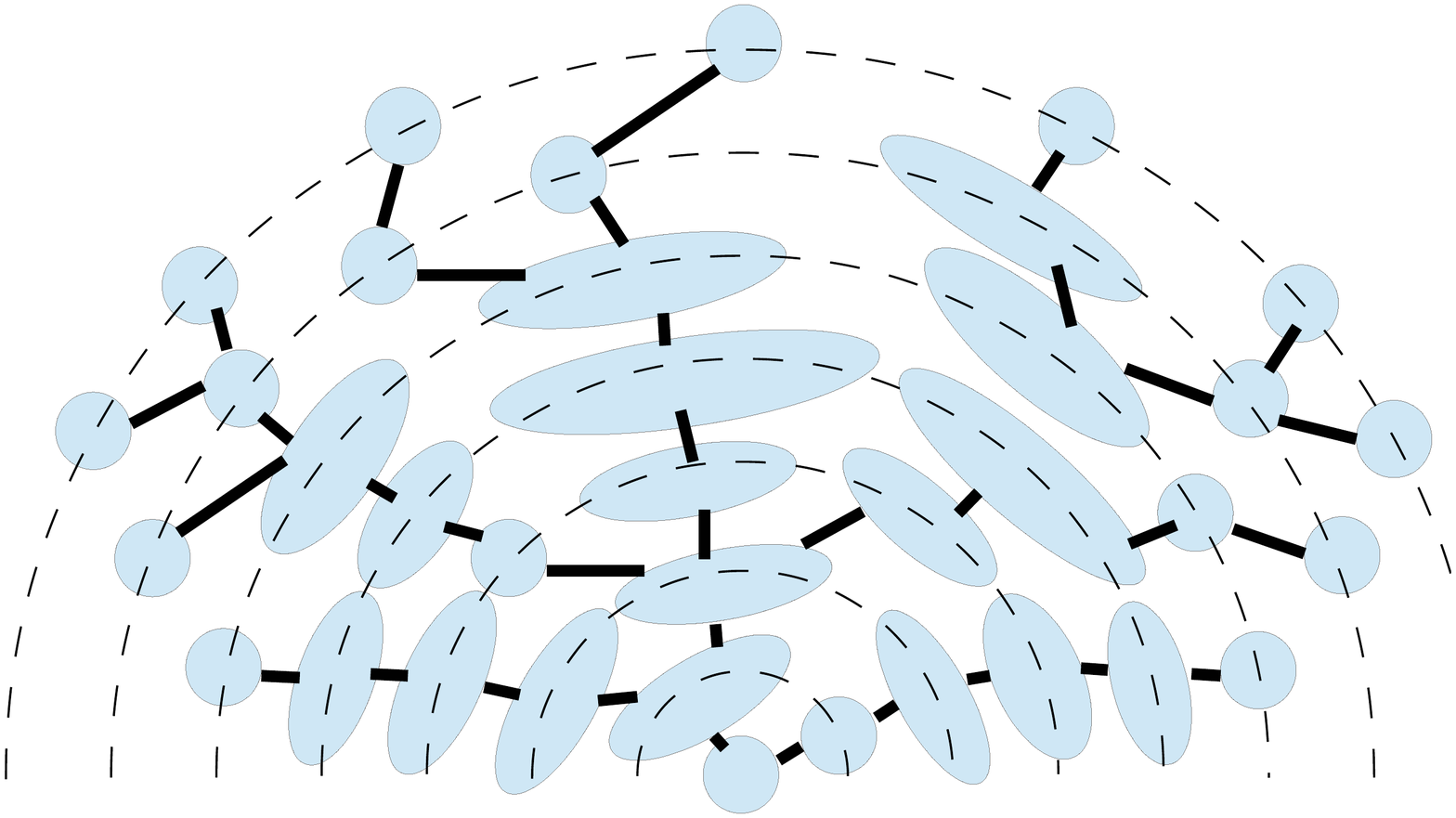}}
			\label{fig:gamma}    
		}
		\caption{\small Layering partition and layering tree}
		\label{fig:layering-partition}
	\end{figure}

\renewcommand{\thefootnote}{\arabic{footnote}}
Finding the cluster-diameter $\Delta_s(G)$ for a given layering partition $\mathcal{LP}(G,s)$ of a graph $G$ requires $O(n m)$ time
\footnote{The parameters $\Delta(G)$ and $\widehat{\Delta}(G)$  can also be computed in total $O(n m)$ time for any graph $G$.}
(we need to know the distance matrix of $G$), although the construction of the layering partition $\mathcal{LP}(G,s)$ itself, for a given vertex $s$, takes only $O(n+m)$ time.

It is not hard to show that, for any graph $G$ and any two of its vertices $s,q$, $\Delta_q(G)\leq 3 \Delta_s(G)$. Thus, the choice of the start vertex for constructing a layering partition of $G$   is not that important.

\begin{proposition} \label{prop:ClustDiamAtAnys}
	Let $s$ be an arbitrary vertex of $G$. For every vertex $q$ of $G$, $\Delta_q(G)\le 3\Delta_s(G)$. In particular,
	$\widehat{\Delta}(G)\le 3 \Delta(G)$ for every graph $G$.
\end{proposition}
\begin{proof}
	Let $\mathcal{LP}(G,s)$ be a layering partition and let $\Gamma(G,s)$ be the corresponding layering tree. Consider an arbitrary vertex $q$ and let $x,y$ be  two vertices of a cluster of $\mathcal{LP}(G,q)$ such that $d_G(x,y)=\Delta_q(G)$. If both $x,y$ belong to the same cluster in $\mathcal{LP}(G,s)$, then  $d_G(x,y)\leq\Delta_s(G)$. Thus, assume that $x,y$ are in different nodes of $\Gamma(G,s)$. Let $X,Y,Q$ be the nodes of $\Gamma(G,s)$ containing vertices $x,y,q,$ respectively. Consider the median node $M$ of nodes $X,Y,Q$  in tree $\Gamma(G,s)$ (see Figure \ref{fig:MedianNode_ClustDiam} for an illustration), i.e., the unique node common to three paths of tree $\Gamma(G,s)$ connecting corresponding nodes from  $\{X,Y,Q\}$. We will show that $d_G(y,M)$ is at most $\Delta_s(G)$ (the same holds for $d_G(x,M)$) and since cluster $M$ has  diameter at most $\Delta_s(G)$ in $G$, we will get that $d_G(x,y)$ is at most $3\Delta_s(G)$.
	
	By the definition of layering partition $\mathcal{LP}(G,q)$, there exists a path $P(x,y)$ in $G$ (not necessarily shortest) between $x$ and $y$  such that no $w\in P(x,y)$ with $d_G(q,w) < d_G(q,y)$ exists. Let $w$ be a vertex of $M\cap P(x,y)$ and $y'$ be a vertex from $M$ laying on a shortest path of $G$ between $y$ and $q$ (such vertices necessarily exist  by the construction of tree $\Gamma(G,s)$ and choice of $M$; $M$ separates in $G$ $y$ from $x$ and $q$, or it contains some of those vertices). We know $d_G(q,y) \leq d_G(q,w)$.
	
	Hence, $d_G(q,y')+d_G(y,y')=d_G(q,y) \leq d_G(q,w) \leq d_G(q,y')+d_G(y',w),$ i.e., $d_G(y,y') \leq d_G(y',w).$ As $d_G(y',w)\leq \Delta_s(G)$, we get
	$d_G(y,M)\leq d_G(y,y')\leq d_G(y',w)\leq \Delta_s(G).$
\end{proof}
To see that the result is sharp, consider a graph in Figure \ref{fig:ClustDiamAtAny}. For that graph and its vertices $q$ and $s$,  $\Delta_q(G)=6$ and $\Delta_s(G)=2$ hold. Furthermore, subdividing every edge $k-1$ times yields a graph $G$ for which $\Delta_q(G)=6k$ and $\Delta_s(G)=2k$ for every integer $k\ge 1$.

  \begin{figure}[htbp]\footnotesize
	\centering                                                        
	\subfigure[][]
	{                    
		\scalebox{0.28}[0.28]{\includegraphics{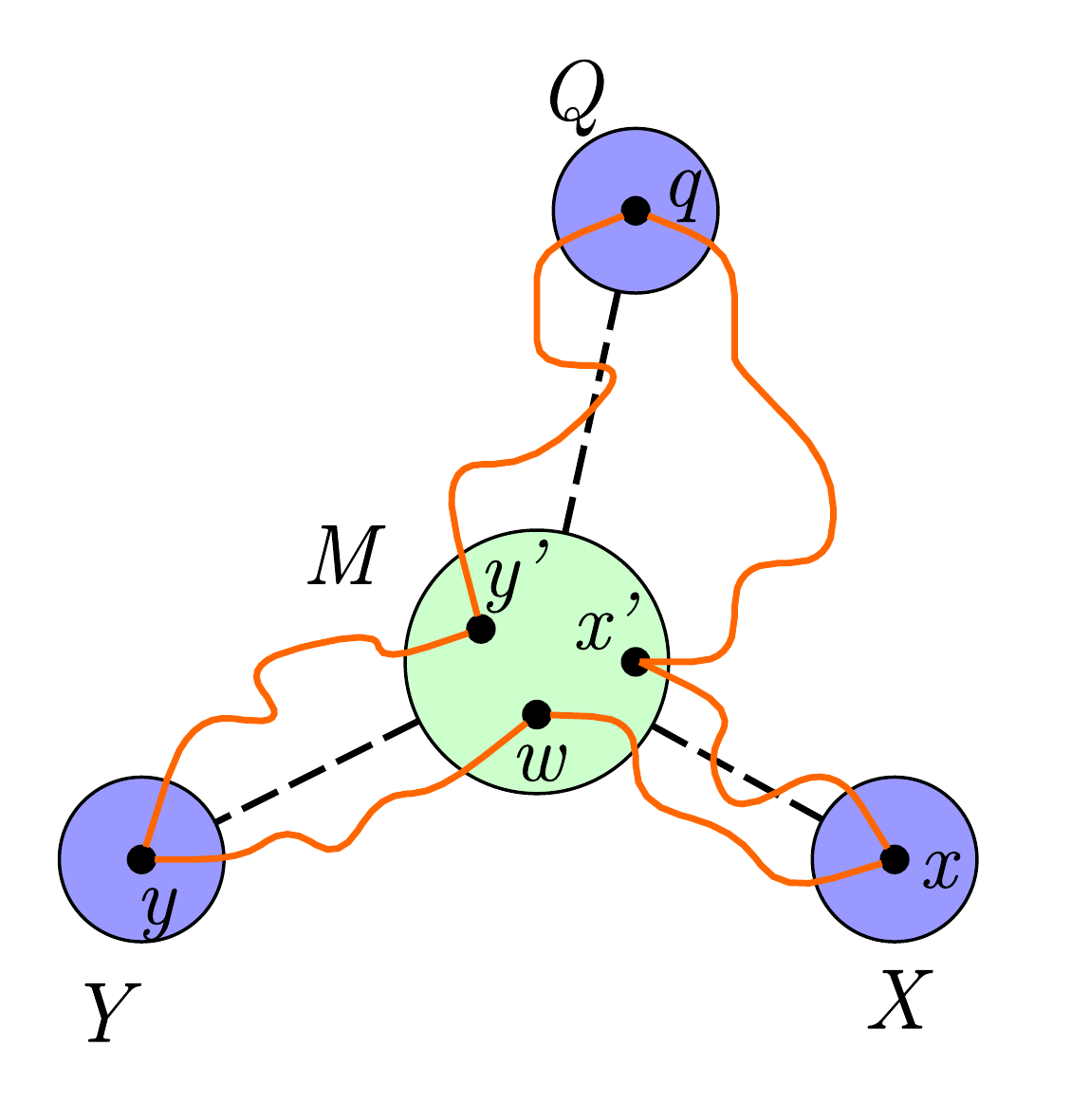}}
		\label{fig:MedianNode_ClustDiam}      
	}
	\hspace{7ex}%
	\subfigure[][]
	{
		\scalebox{0.28}[0.28]{\includegraphics{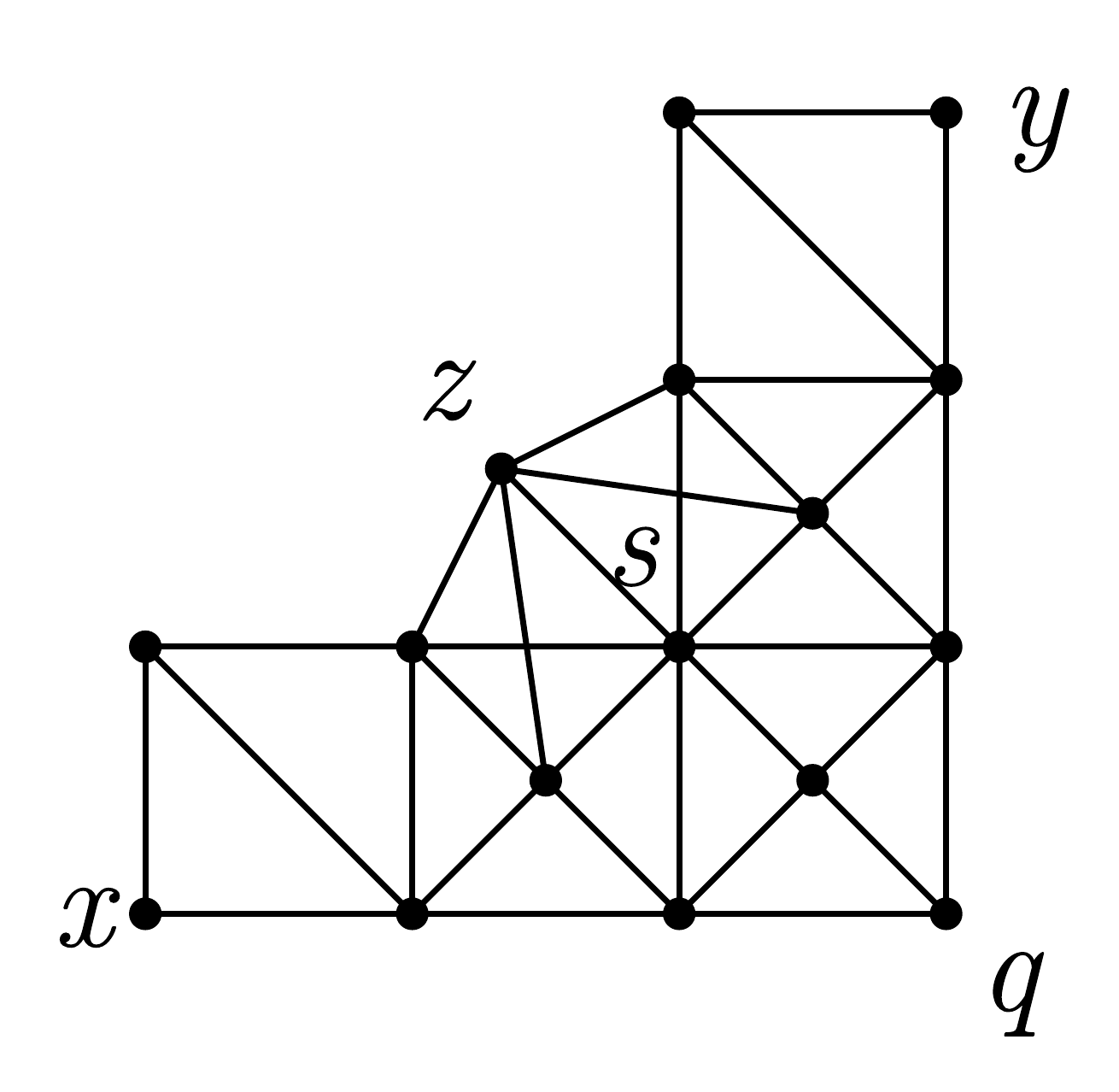}}
		\label{fig:ClustDiamAtAny}
	}
	
	\subfigure[][]
	{                    
		\scalebox{0.28}[0.28]{\includegraphics{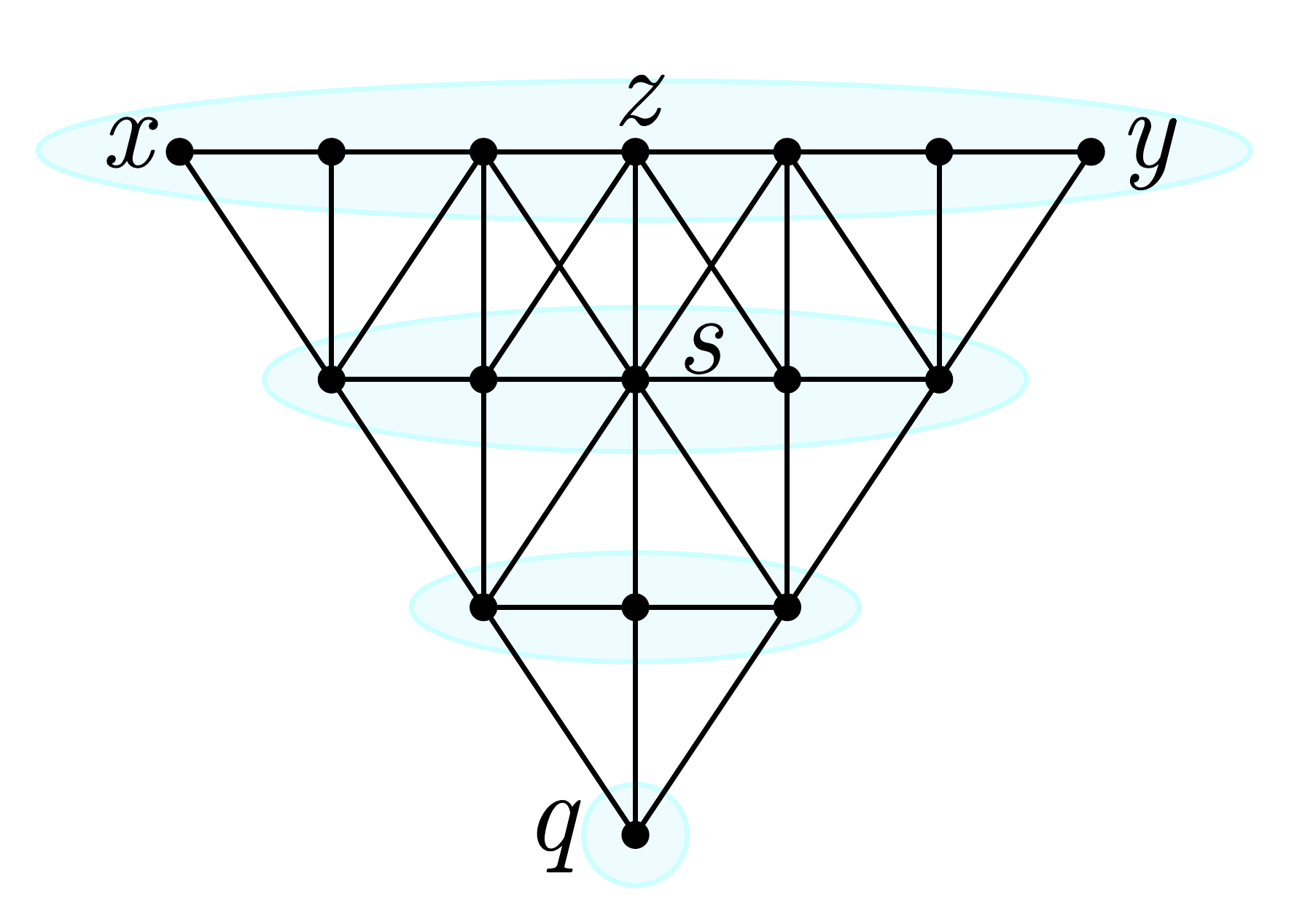}}
		\label{fig:LayeringPartitionAt_q}     
	}
	\subfigure[][]
	{                    
		\scalebox{0.28}[0.28]{\includegraphics{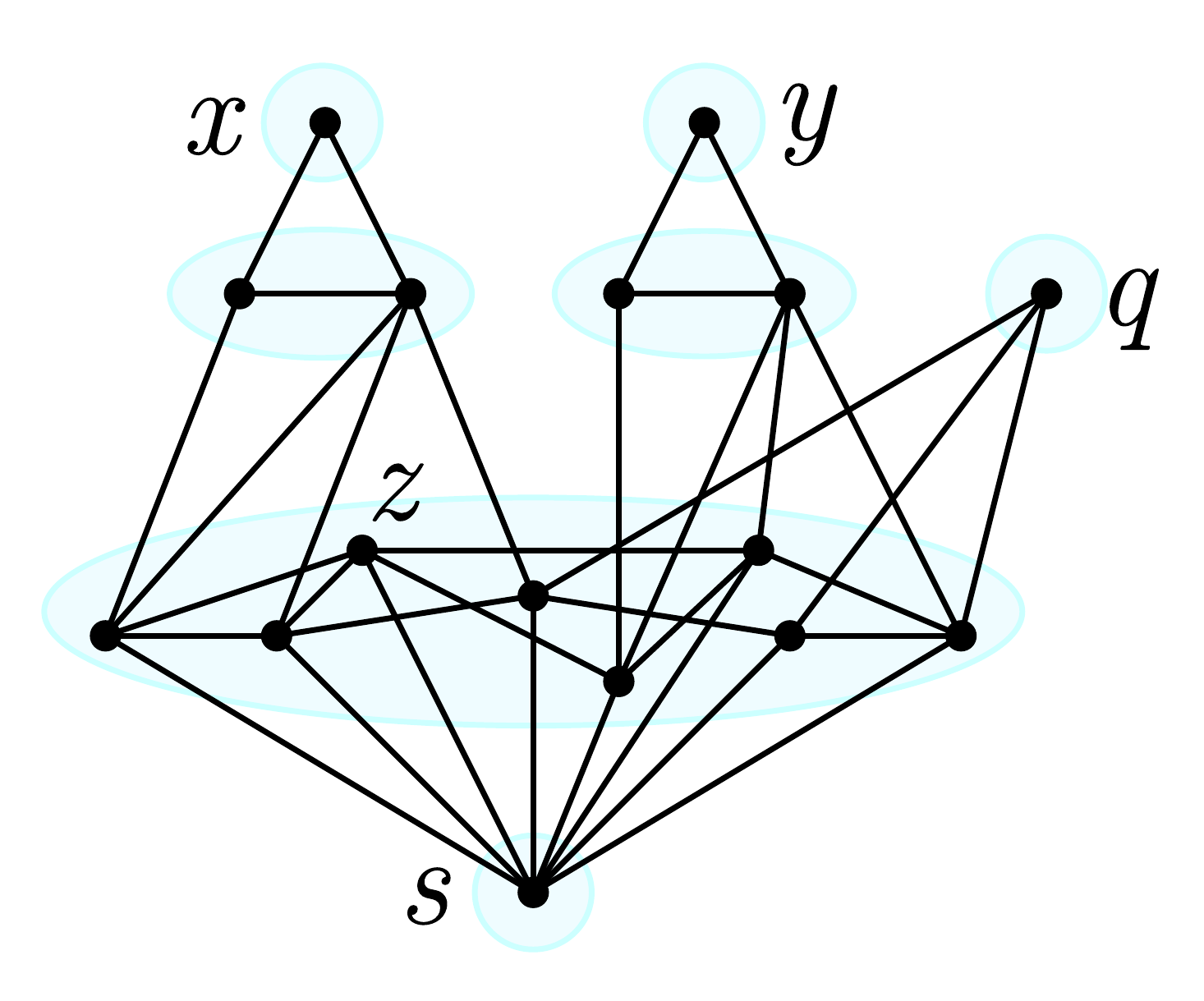}}
		\label{fig:LayeringPartitionAt_s}     
	}
	\vspace{-2ex}
	\caption{\small (a) Illustration to the proof of Proposition \ref{prop:ClustDiamAtAnys}; (b) Sharpness of inequalities in 
		Proposition \ref{prop:ClustDiamAtAnys}; (c) and (d) Clusters of the layering partitions of the graph from 
		Figure \ref{fig:ClustDiamAtAny} with respect to $q$ and $s$, respectively.}
	\label{fig:Prop1}  
\end{figure}

The following theorem establishes a relationship between the slimness and the cluster-diameter of a layering partition of a graph.		

\begin{theorem} \label{TH:Slimness_ClusterDiam}
	For every graph $G$,  $sl(G)\leq\lfloor\frac{\widehat{\Delta}(G)}{2}\rfloor$.
\end{theorem}

\begin{proof}
	Consider any geodesic triangle $\bigtriangleup=\bigtriangleup (x,y,z)$ of $G$ formed by shortest paths $P(x,y)$, $P(x,z)$ and $P(z,y)$ connecting corresponding vertices. Assume that $c$ is an arbitrary vertex of $P(x,y)$ and $\ell \geq 0$ is the maximum integer such that $B_{\ell}(c)$ does not intersect $P(x,z)\cup P(z,y)$.
	
	Let $\mathcal{LP}(G,c)$ be a layering partition of $G$ starting at vertex $c$. Consider two vertices $a$ and $b$ in $P(x,y)$ such that $d_G(c,a)=d_G(c,b)=\ell$, $a\neq b$. Let also $a'$ and $b'$ be vertices in $P(x,y)$ such that $d_G(c,a')=d_G(c,b')=\ell+1$, $a$ is adjacent to $a'$ and $b$ is adjacent to $b'$ (see Figure \ref{fig:BFS_ClusterDiameter} for an illustration). Since $B_{\ell+1}(c)$ intersects $P(x,z)\cup P(z,y)$ but $B_{\ell}(c)$ does not, vertices $a',b'$ are in the same cluster of $\mathcal{LP}(G,c)$. Hence, $d_G(a',b')\le \Delta_c(G)\le \widehat{\Delta}(G)$. On the other hand,
	$d_G(a',b')=1+d_G(a,b)+1=2\ell+2$. Hence, $d_G(c,P(x,z)\cup P(z,y))=\ell+1= \frac{d_G(a',b')}{2}\le \frac{\widehat{\Delta}(G)}{2}$.
\end{proof}

\begin{figure}[htbp]\footnotesize
	\centering                                                        
	\subfigure[][]
	{                    
		\scalebox{0.28}[0.28]{\includegraphics{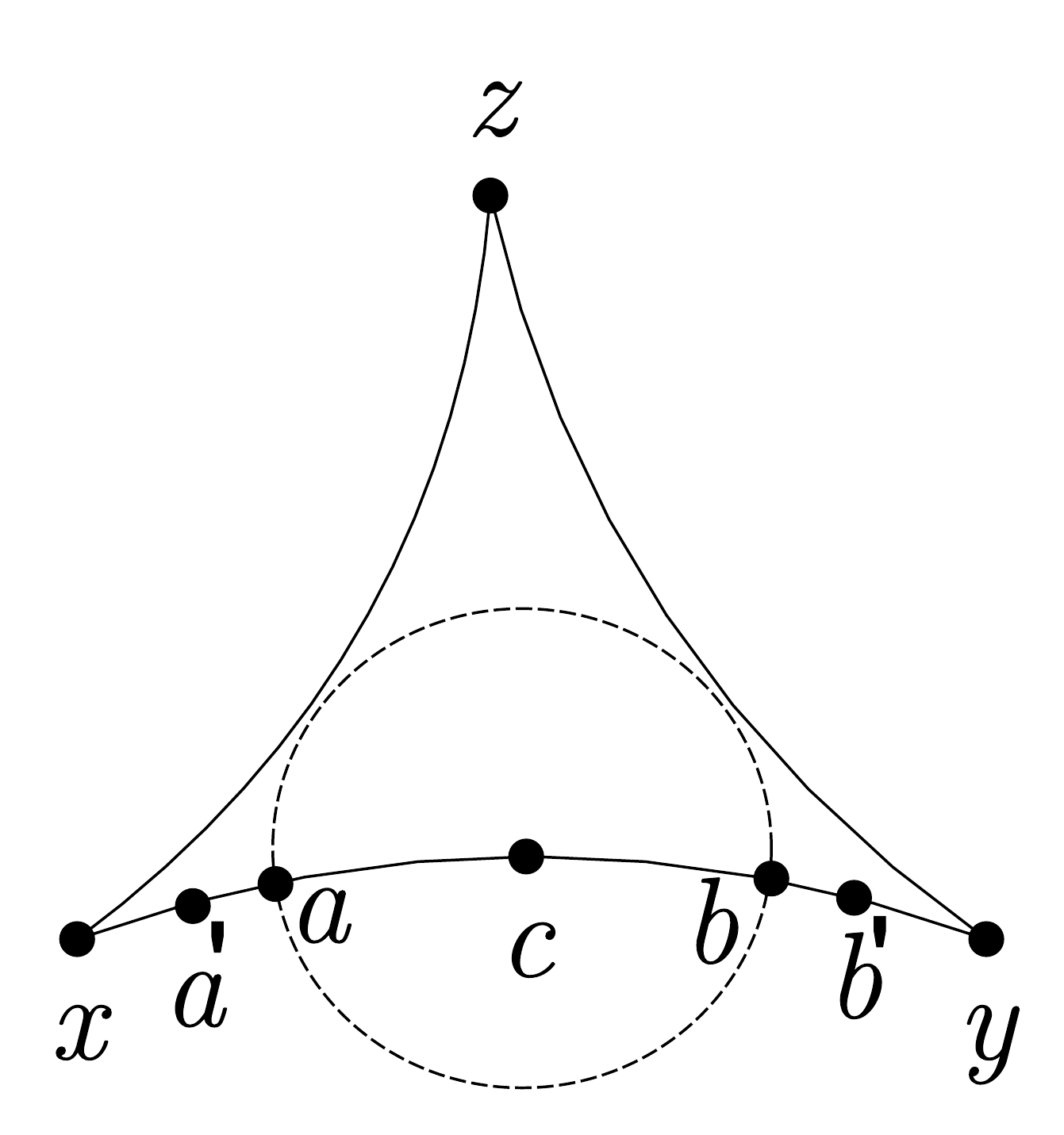}}
	 	
		\label{fig:ClusterDiameter_Slim}     
	}
	\hspace{7ex}%
	\subfigure[][]
	{
		\scalebox{0.28}[0.28]{\includegraphics{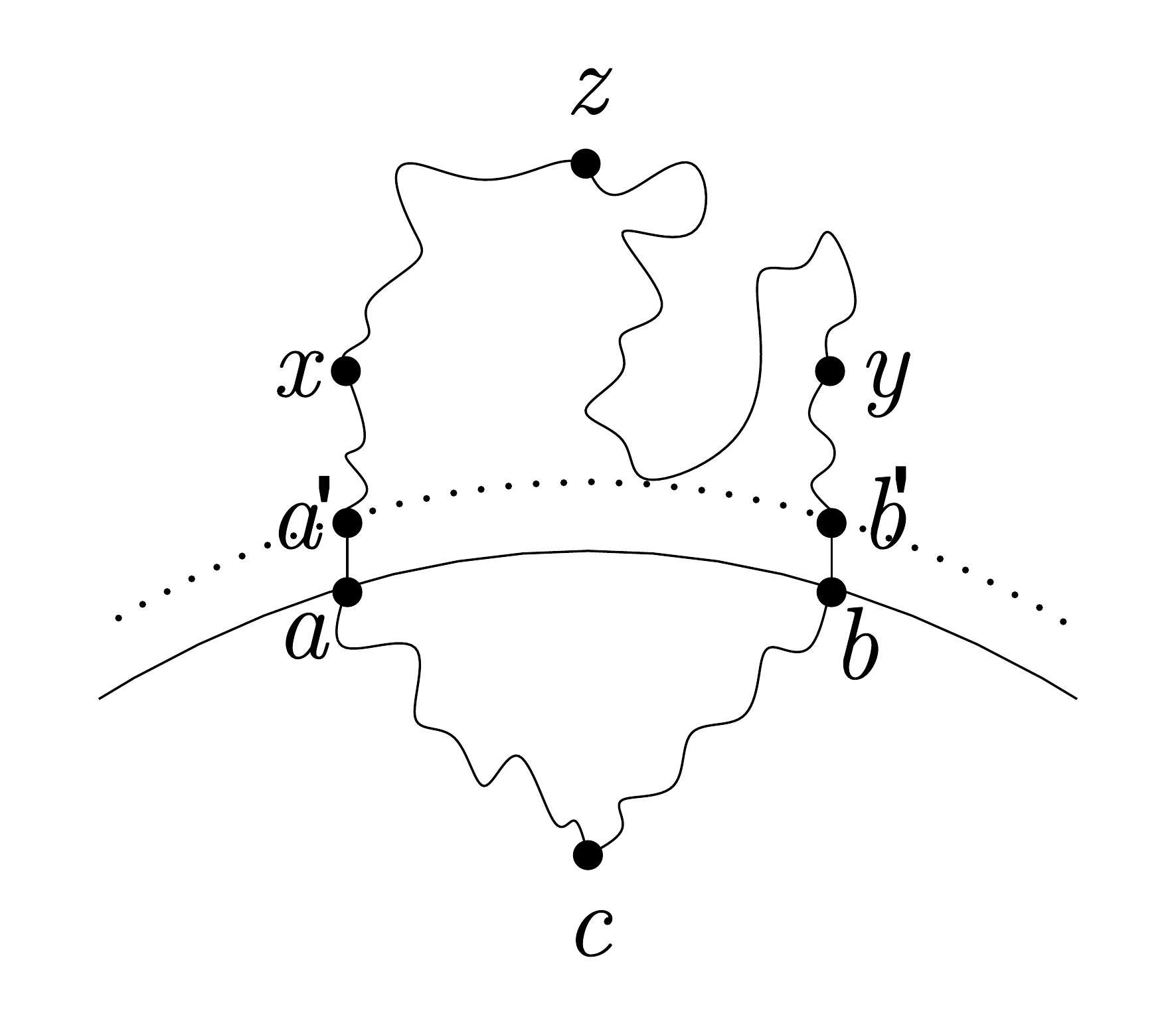}}
		
		\label{fig:LayeringPartition}
	}
	\caption{\small Illustration to the proof of Theorem~\ref{TH:Slimness_ClusterDiam}. Geodesic triangle $\triangle(x,y,z)$ and the Layering Partition starting with vertex $c$. }
	\label{fig:BFS_ClusterDiameter} 
\end{figure}

To see that the result is sharp, consider the graph in Figure \ref{fig:Slim_Chordality_Sharpness}.
In a geodesic triangle $\bigtriangleup (x,y,z)$ of $G$, formed by shortest paths $P(x,y)=(x,9,5,2,c,1,3,6,y)$, $P(x,z)=(x,15,11,14,z)$ and $P(z,y)=(z,13,10,12,y)$, vertex $c\in P(x, y)$ is at distance 4 from $P(x, z) \cup P(y, z)$.
Hence, the slimness of $G$ is at least 4. In a layering partition starting at $c$, all three vertices $x,y,z$ belong to the same cluster. Therefore, $\Delta_c(G)=8$ and, since $diam(G)=8$ as well, $\widehat{\Delta}(G)=8$.  That is, $sl(G)=4=\frac{\widehat{\Delta}(G)}{2}$. Also, subdividing every edge $k-1$ times yields a graph $G$ for which $sl(G)=4k=\frac{\widehat{\Delta}(G)}{2}$ for every integer $k\ge 1$.

Combining Proposition \ref{prop:ClustDiamAtAnys} with Theorem \ref{TH:Slimness_ClusterDiam}, one gets the following corollary.

\begin{corollary} \label{cor:ineq}
	For every graph $G$ and every vertex $s$ of $G$, $sl(G)\leq \lfloor\frac{3}{2}\Delta_s(G)\rfloor$. 	
\end{corollary}

To see that the inequality  in Corollary \ref{cor:ineq} is sharp, consider again the graph in Figure \ref{fig:ClustDiamAtAny}.
We have $\Delta_s(G)=2$ and $sl(G)\ge 3$ (check geodesic triangle $\bigtriangleup (x,y,z)$, formed by shortest paths $P(x,z)$ (the north-west one), $P(y,z)$ (the north-west one) and $P(x,y)$ (the south-east one), and vertex $q$ in $P(x,y)$). Furthermore, subdividing every edge $k-1$ times yields a graph $G$ for which $\Delta_s(G)=2k$ and $sl(G)\ge 3k$ for every integer $k\ge 1$. Hence, $sl(G)=3k=\frac{3}{2}\Delta_s(G)$.

\begin{figure}[htbp]
	\centering
	\includegraphics[scale=0.3]{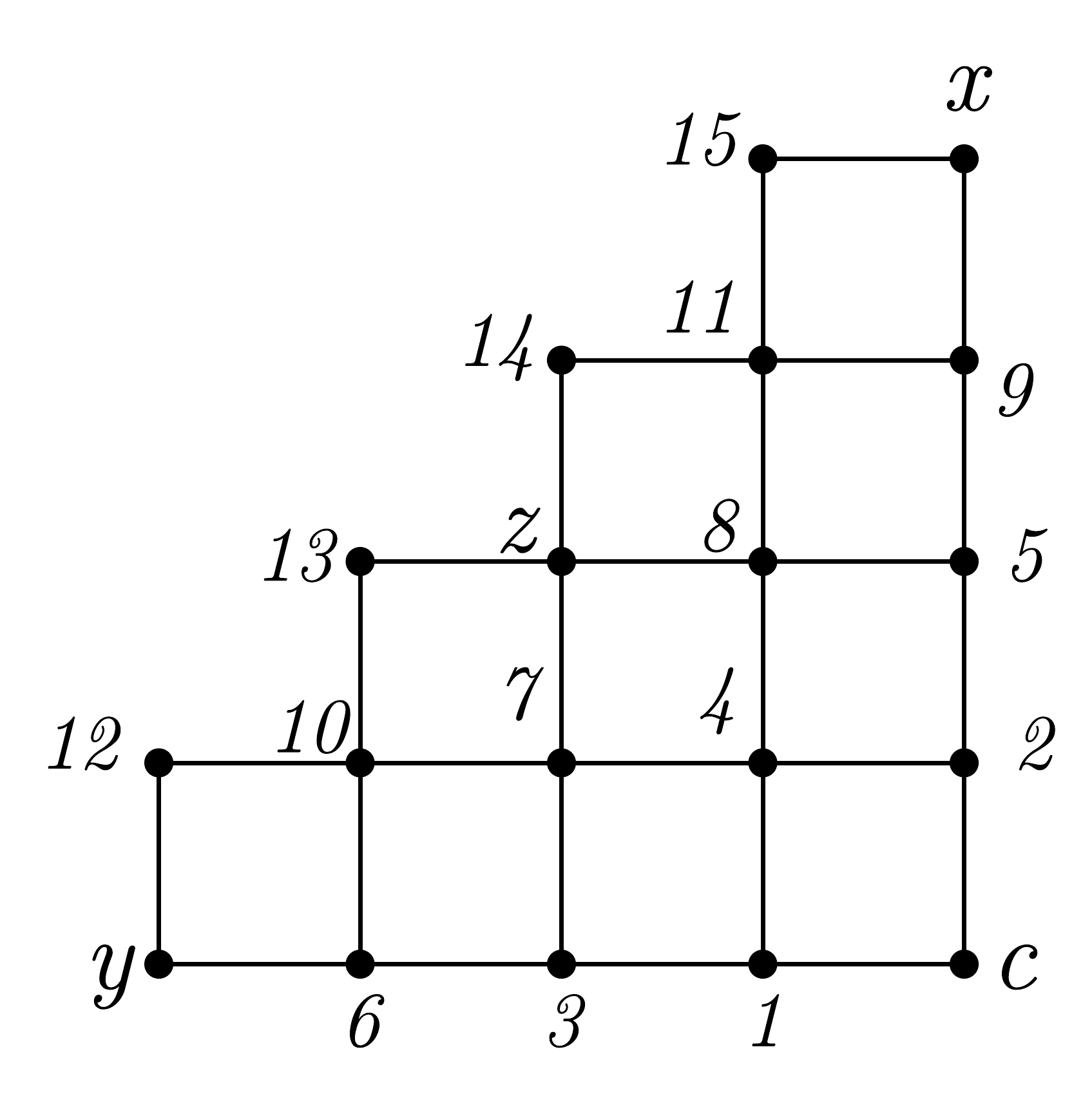}
	\caption{Sharpness of the bounds for slimness through cluster-diameter and chordality.}
	\label{fig:Slim_Chordality_Sharpness}
\end{figure}

\subsection{Slimness and tree-length}
Our next graph parameter is based on the notion of tree-decomposition introduced by Robertson and Seymour in their work
on graph minors~\cite{RobSey86}.

A \emph{tree-decomposition} of a graph $G=(V,E)$ is a  pair
$(\{X_i: i\in I\},T=(I,F))$ where $\{X_i: i\in I\}$ is a collection of
subsets of $V$, called \emph{bags}, and $T$ is a tree. The nodes of
$T$ are the bags $\{X_i: i\in I\}$ satisfying the following three
conditions: \\
1) $\bigcup_{i\in I}X_i=V$;\\
2) for each edge $uv\in E$, there is a bag $X_i$ such that $u,v \in
X_i$;\\
3) for all $i,j,k \in I$, if $j$ is on the path from $i$ to $k$ in $T$, then $X_i \cap X_k\subseteq X_j$.\\
For simplicity we denote a tree-decomposition $\left(\{X_i: i\in
I\},T=(I,F)\right)$ of a graph $G$ by $\mathcal{T}(G)$.

The \emph{width} of a tree-decomposition $\mathcal{T}(G)=(\{X_i: i\in
I\},T=(I,F))$ is $\max_{i\in I}|X_i|-1$. The \emph{tree-width} of a
graph $G$, denoted by $tw(G)$, is the minimum width over all
tree-decompositions $\mathcal{T}(G)$ of $G$~\cite{RobSey86}.

The \emph{length} of a
tree-decomposition $\mathcal{T}(G)$ of a graph $G$ is $\lambda:=\max_{i\in
	I}\max_{u,v\in X_i}d_G(u,v)$ (i.e., each bag $X_i$ has diameter at
most $\lambda$ in $G$). The \emph{ tree-length} of $G$, denoted by
$tl(G)$, is the minimum length  over all tree-decompositions
of $G$~\cite{DoGa2007}. The chordal graphs  are exactly the graphs
with tree-length 1. Note that these two graph parameters are not
directly related to each other on general graphs. For instance, a clique
on $n$ vertices has tree-length 1 and tree-width $n-1$,
whereas a cycle on $3n$ vertices has tree-width 2 and tree-length
$n$. However, if one involves also the size $\ell(G)$ of a largest isometric cycle in $G$, then a relation is possible:
$tl(G)\leq \lfloor\ell(G)/2\rfloor(tw(G)-1)$ \cite{coudert2016teelengthapproximate} (see also \cite{adcock2016treedecompositions} and \cite{diestel2012connectedTW} for similar but slightly weaker bounds: $tl(G)\leq \ell(G)(tw(G)+1)$ \cite{adcock2016treedecompositions}; $tl(G)\leq \ell(G)(tw(G)-2)$ \cite{diestel2012connectedTW}). Furthermore, the tree-width of a planar graph $G$ is bounded by O(tl(G))~\cite{DiengG09}.

The \emph{breadth} of a tree-decomposition $\mathcal{T}(G)$ of a graph $G$ is the
minimum integer $r$ such that for every $i\in I$ there is a vertex
$v_i\in V$ with $X_i\subseteq B_r(v_i,G)$ (i.e., each bag $X_i$
can be covered by a ball $B_r(v_i,G)$ of radius at most $r$ in $G$).
The \emph{tree-breadth} of $G$, denoted by
$tb(G)$, is the minimum breadth over all
tree-decompositions of $G$ ~\cite{DBLP:conf/approx/DraganK11}. Evidently, for any graph $G$, $1\leq
tb(G)\leq tl(G)\leq 2 tb(G)$ holds.

Unfortunately, while graphs with tree-length 1 (as they are
exactly the chordal graphs) can be recognized in linear time, the
problem of determining whether a given graph has tree-length at most
$\lambda$ is NP-complete for every fixed $\lambda >1$ (see~\cite{Daniel10}). Judging from this result, it was conceivable that the
problem of determining whether a given graph has tree-breadth at most $\rho$ is NP-complete, too. Indeed, it was confirmed recently that the problem is NP-complete for every $\rho \geq 1$ \cite{DBLP:conf/iwoca/DucoffeLN16}. On the positive side, both parameters $tl(G)$ and $tb(G)$ can be easily approximated within a factor of 3 using a layering partition of $G$ \cite{DoGa2007,DBLP:conf/approx/DraganK11,DBLP:journals/networks/Abu-AtaD16}.

The following proposition establishes a relationship between the tree-length and the cluster-diameter of a layering partition of a graph.
\begin{proposition} {\cite{DoGa2007}}\label{prop:dorisb}
	For every graph $G$ and any vertex $s$,
	$\Delta_s(G)/3 \leq tl(G) \leq \Delta_s(G)+1.$
\end{proposition}
Thus, the cluster-diameter $\Delta_s(G)$ of a layering partition provides easily computable bounds for the hard to compute parameter $tl(G)$.

Diestel and M\"uller \cite{diestel2012connectedTW} proved that every graph $G$  of tree-length $tl(G)$ has slimness at most $\lfloor\frac{3}{2}tl(G)\rfloor$. Using Theorem \ref{TH:Slimness_ClusterDiam} and Proposition \ref{prop:dorisb} by Dourisboure and Gavoille \cite{DoGa2007}, we obtain that result 
as a simple corollary.

\begin{corollary}{\cite{diestel2012connectedTW}} \label{cor:DiestelM}
	For every graph $G$, $sl(G)\leq\lfloor\frac{3}{2}tl(G)\rfloor$.
	\label{Corolllary:Slim_n_Tree-length}
\end{corollary}

Diestel and M\"uller in \cite{diestel2012connectedTW} showed also that the bound in Corollary \ref{cor:DiestelM} is sharp; for every integer $k$ there is a graph $G$ with tree-length $k$ such that $sl(G)=\lfloor\frac{3}{2}k\rfloor$.

\commentout{	
	To show that the bound of $\lfloor\frac{3}{2}tl(G)\rfloor$ is sharp, consider a path-decomposition of adhesion 2 into many copies of $K^4$ and with disjoint adhesion sets. Subdividing every edge $tl(G) - 1$ times yields a graph $G$ of tree-length at most $tl(G)$: the tree-decomposition witnessing this is a path-decomposition whose parts are the 4-vertex-sets from the $K^4$ s together with pendent parts each consisting of a subdivided edge.
	
	It is easy to find in $G$ three vertices $x, y, z$ with shortest paths $P(x, y), P(y, z)$ and $P(x, z)$ between them such that each subdivided $K^4$ (other than the first and the last) meets $P(x, z)$ in exactly one subdivided edge $P$ and $P(x, y) \cup P(y, z)$ in exactly one subdivided edge $Q$ disjoint from $P$ (Figure \ref{fig:Slim_TreeLength_Sharpness}). For each subdivided $K^4$ except the leftmost and the rightmost one, the vertex closest to the mid-point of $P \subseteq P(x, z)$ has distance exactly $\lfloor\frac{3}{2}tl(G)\rfloor$ not only from $Q$ but from all of $P(x, y) \cup P(y, z)$.
	
	\begin{figure}[htbp]
		\centering
		\includegraphics[scale=0.3]{figures/SlimnessBound/Slimness_TreeLength}
		\caption{To show the sharpness of the bound for slimness with respect to tree-length.}
		\label{fig:Slim_TreeLength_Sharpness}
	\end{figure}
}

As for every graph $G$, $tl(G)\leq 2tb(G)$, we have also the following inequality.

\begin{corollary} \label{cor:sl-vs-tb}
	For every graph $G$, $sl(G)\leq 3~tb(G)$.
\end{corollary}

\begin{figure}[htbp]
	\centering
	\includegraphics[scale=0.22]{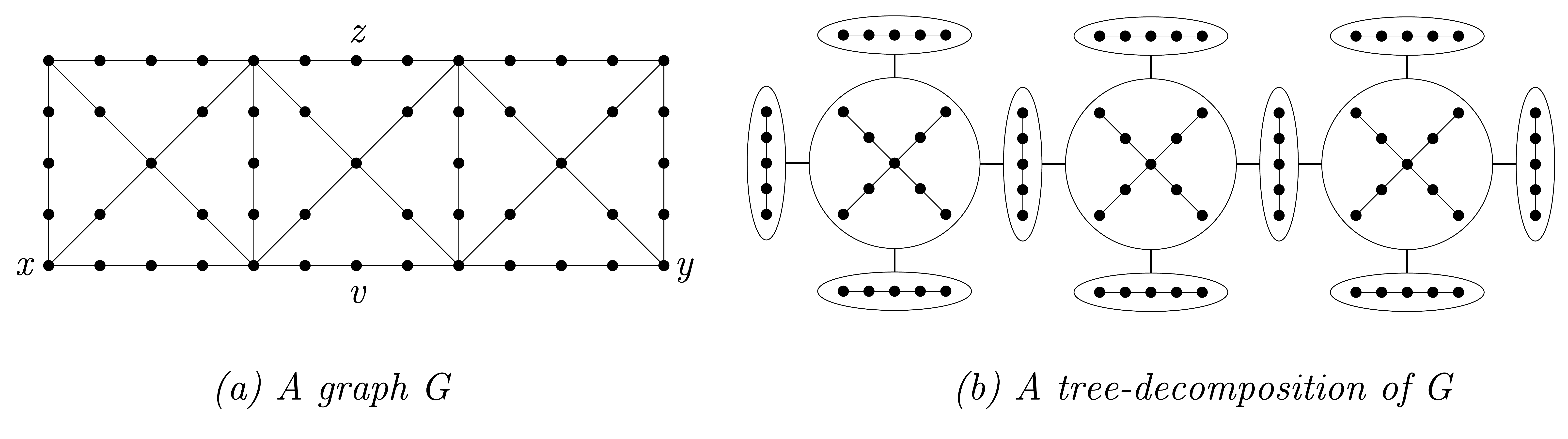}
	\caption{Sharpness of the bound for slimness with respect to tree-breadth.}
	\label{fig:Slim_TreeBreadth_Sharpness}
\end{figure}

The inequality  in Corollary \ref{cor:sl-vs-tb} is also sharp. Consider the graph $G$ in Figure \ref{fig:Slim_TreeBreadth_Sharpness}. The tree-breadth of this graph is at most 2 as witnessed by the shown tree-decomposition. The slimness of $G$ is at least 6. Indeed, consider a geodesic triangle $\bigtriangleup (x,y,v)$, formed by two horizontal shortest paths $P(x,v)$ and $P(y,v)$ and a shortest path $P(x,y)$ containing vertex $z$. Then, vertex $z$ of $P(x,y)$ is at distance 6 from $P(x,v)\cup P(y,v)$. Thus, $6\leq sl(G)\leq 3 tb(G)\leq 6$, i.e., $sl(G)= 3~tb(G)= 6$.
Clearly, this graph can be modified to get a graph $G$ such that $tb(G)= k$ and $sl(G)= 3k$ (contract some edges to get $k=1$; subdivide some edges to get $k\ge 3$).

\section{Graphs with (small) chordality}
A graph $G$ is $k$-chordal if every induced cycle of $G$ has length at most $k$. The parameter $k$ is usually called \emph{the chordality} of $G$.
When $k=3$, $G$ is called a \emph{chordal graph}.

Using Theorem \ref{TH:Slimness_ClusterDiam}, we can also obtain a sharp bound on the slimness through the chordality of a graph. In \cite{DBLP:journals/ejc/ChepoiD00}, Chepoi and Dragan proved that for every $k$-chordal graph $G$ and every its vertex $s$, $\Delta_s(G)\leq k/2+2$. Combining this with Theorem \ref{TH:Slimness_ClusterDiam}, we obtain a result of Bermudo et al. \cite{bermudo2016hyperbolicity} as a simple corollary.

\begin{corollary}{ \cite{bermudo2016hyperbolicity}} \label{cor:k-ch}
	For every $k$-chordal graph $G$, $sl(G)\leq\lfloor k/4\rfloor+1$.
\end{corollary}

To see that the result is sharp, consider again the graph from Figure \ref{fig:Slim_Chordality_Sharpness}. We know that $sl(G)=4$.   The chordality of $G$ is 12 as witnessed by an induced cycle $(c,2,5,9,11,14,z,13,10,6,3,1)$ of length 12. Considering (half) rectilinear grids with longer or shorter sides, we obtain for every integer $k\ge 1$ a graph $G$ with chordality $4k$ and slimness $k+1$.

From Corollary \ref{cor:k-ch} we know that every chordal graph has slimness at most 1 and every $k$-chordal graph with $k\le 7$  has slimness at most 2.  In this section, we also show that the graphs with slimness equal 0 are exactly the block graphs and an interesting subclass of 5-chordal graphs, namely the class of AT-free graphs, has slimness at most 1. Furthermore, we characterize those 4-chordal graphs for which every induced subgraph has slimness at most 1. As a consequence, we get that every house-hole-domino--free graph (HHD-free graph) has slimness at most 1.

\subsection{AT-free graphs and block graphs}
An independent set of three vertices such that each pair is joined by a path that avoids the neighborhood of the third is called an \emph{asteroidal triple}. A graph $G$ is an \emph{AT-free graph} if it does not contain any asteroidal triples. The class of AT-free graphs contains many intersection families of graphs, including permutation graphs, trapezoid graphs, co-comparability graphs and interval graphs.

\begin{proposition} 	\label{prop:AT-free_N_slim}
	Every AT-free graph has slimness at most 1.
\end{proposition}

\begin{proof} Consider any geodesic triangle $\bigtriangleup=\bigtriangleup (x,y,z)$ of $G$ formed by shortest paths $P(x,y)$, $P(x,z)$ and $P(z,y)$ connecting corresponding vertices. Assume there is a vertex $u$ in  $P(x,y)$ such that $B_{1}(u)$ does not intersect $P(x,z)\cup P(z,y)$. Then, $x,y,u$ form an independent set in $G$ and there is a path in $P(x,z)\cup P(z,y)$ that connects vertices $x$ and $y$ and avoids the neighborhood of $u$. Furthermore, as $P(x,y)$ is a shortest path, its subpath $P(x,u)$ (subpath $P(y,u)$) avoids the neighborhood of $y$ (of $x$, respectively).
	Hence, vertices $x,u,y$ form an asteroidal triple in $G$, contradicting with $G$ being an AT-free graph.
\end{proof}

%
A block graph is a connected graph whose blocks (i.e., maximal biconnected subgraphs) are cliques.  A diamond $K_4-e$ is a  complete graph on 4 vertices minus one edge. We will need the following characterization of block graphs.

\begin{proposition}{\cite{bandelt1986distanceHGraphs}}
	\label{prop:block_graph}
	$G$ is a block graph if and only if $G$ has neither diamonds nor cycles $C_k$ of length $k\geq 4$ as isometric subgraphs.
\end{proposition}

\begin{proposition} \label{prop:slim_N_bloch}
	$sl(G)=0$ if and only if $G$ is a block graph.
\end{proposition}
\begin{proof} Let $G$ be a block graph. Since $hb(G)=0$ for a block graph, it follows immediately from the inequality $sl(G) \leq 3hb(G)+ \frac{1}{2}$ and the integrality of the slimness that $sl(G)=0$. Conversely, if slimness of a graph is 0 then, for every geodesic triangle $\bigtriangleup(x,y,z)$, $P(x,y)\subseteq P(x,z)\cup P(z,y)$ must hold. So, a graph $G$ with slimness 0 cannot have a diamond or any cycles $C_k$ of length $k\geq 4$ as an isometric subgraph. Hence, by Proposition  \ref{prop:block_graph}, $G$ is a block graph.
\end{proof}

\subsection{4-chordal graphs}
We know that all 4-chordal graphs have slimness at most 2. In what follows, we characterize those 4-chordal graphs for which every induced subgraph has slimness at most 1. As a corollary, we get that every house-hole-domino--free graph (HHD-free graph) has slimness at most 1. A \emph{chord} in a cycle is an edge connecting two non-consecutive vertices of the cycle. Let $C_k$ denote an induced cycle of length $k$.

The following two lemmata will be frequently used in what follows.

\begin{lemma}{(Cycle Lemma for 4-chordal graphs) \cite{dragan1999convexity}} \label{lm:cycle_lemma_H-Free}
	Let $C$ be a cycle of length at least $5$ in a $4$-chordal  graph $G=(V,E)$. Then, for each edge $xy\in C$ there are vertices $w_1, w_2$ in $C$ such that $xw_1\in E$, $yw_2\in E$, and $d_G(w_1, w_2)\leq 1$, i.e., each edge of a cycle is contained in a triangle or a $4$-cycle.
\end{lemma}

A \emph{building} B$(w|uv)$ is a chain of $k$ $C_4$s ($k\ge 0$) ending with a $C_3$ depicted in Figure \ref{fig:multistory_building}.
\begin{lemma}{ (building)} \label{lm:multistory_building}	
	In a $4$-chordal graph $G=(V,E)$, if an edge $uv\in E$ is equidistant from a vertex $x$, that is, $d_G(u,x)=d_G(v,x)$, then
	$G$ contains a  building B$(w|uv)$ as an isometric subgraph  (see Figure \ref{fig:multistory_building}), where $w\in I(x,u)\cap I(x,v)$.
\end{lemma}
\begin{proof} We proceed by induction on $k=d_G(x,u)=d_G(x,v)$. Let $P(x,u)$, $P(x,v)$ be two shortest paths connecting $u$ and $v$ with $x$. Note that the union of $P(x,u)$, $P(x,v)$ and the edge $uv$ contains a cycle. Applying the Cycle Lemma to this cycle and its edge $uv$, by distance requirements, we obtain either a vertex $w$ adjacent to both $u$ and $v$ and at distance $k-1$ from $x$, or two adjacent vertices $u'$ and $v'$ such that $uu',vv' \in E$,
	$uv',vu' \notin E$, and $d(u',x)=d(v',x)=k-1$. In the former case, $w,u,v$ form a triangle and we are done. In the latter case, by the induction  hypothesis, there exists in $G$ an isometric  building B$(w|u'v')$ with $w\in I(x,u')\cap I(x,v')$. Adding to B$(w|u'v')$ vertices $u,v$ and edges $u'u,uv,vv'$ we obtain an isometric  building B$(w|uv)$ with $w\in I(x,u)\cap I(x,v)$.
\end{proof}

\begin{figure}    [htbp]
	\centering
	\includegraphics[scale=0.35]{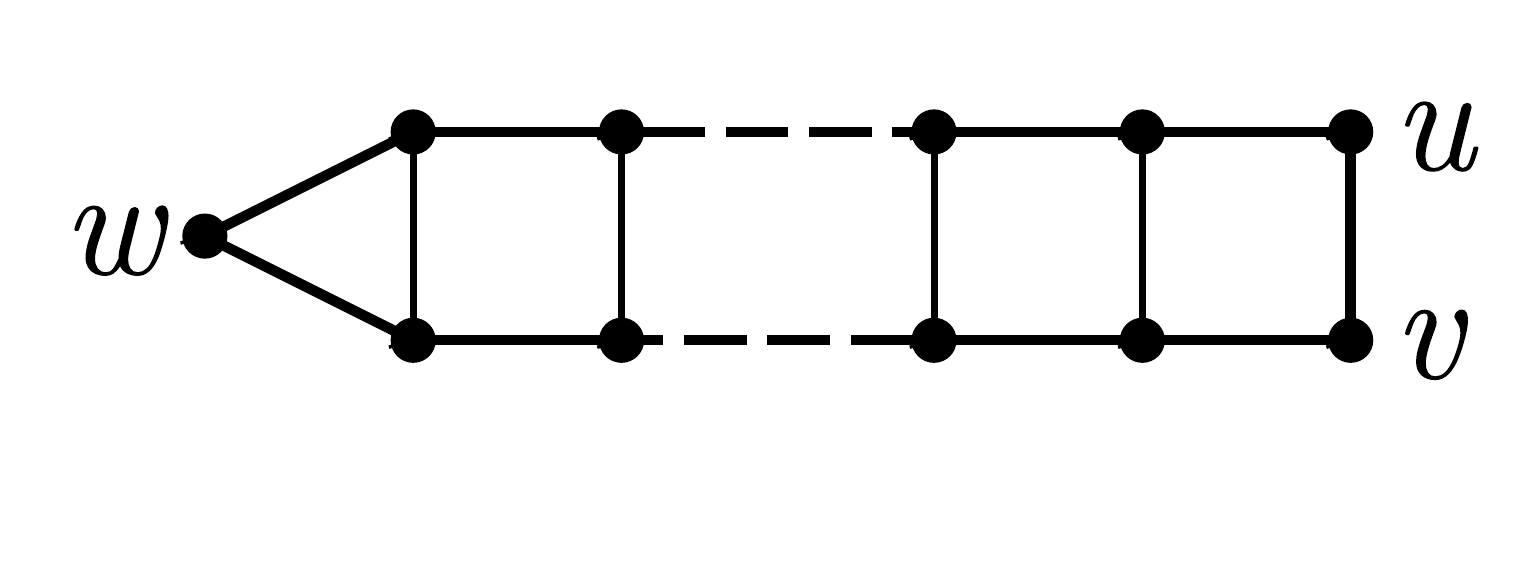}
	\caption{A  building B$(w|uv)$: a chain of $k$ $C_4$s ($k\ge 0$) ending with a $C_3$.}
	\label{fig:multistory_building}
\end{figure}

A  building with $k=1$ is known as a {\sl house}. 	

\begin{theorem} \label{TH:4-chordal_graph}
	For every 4-chordal graph $G$  the following two statements are equivalent.
	\begin{enumerate}[label=(\roman*)]
		\item $G$ and each its induced subgraph has slimness at most 1.
		\item None of the following graphs (see Figure \ref{fig:eightforbidden4Chordal}) is an induced subgraph of $G$.
	\end{enumerate}
	
	\begin{figure}    [H]
		\centering
		\includegraphics[scale=0.21]{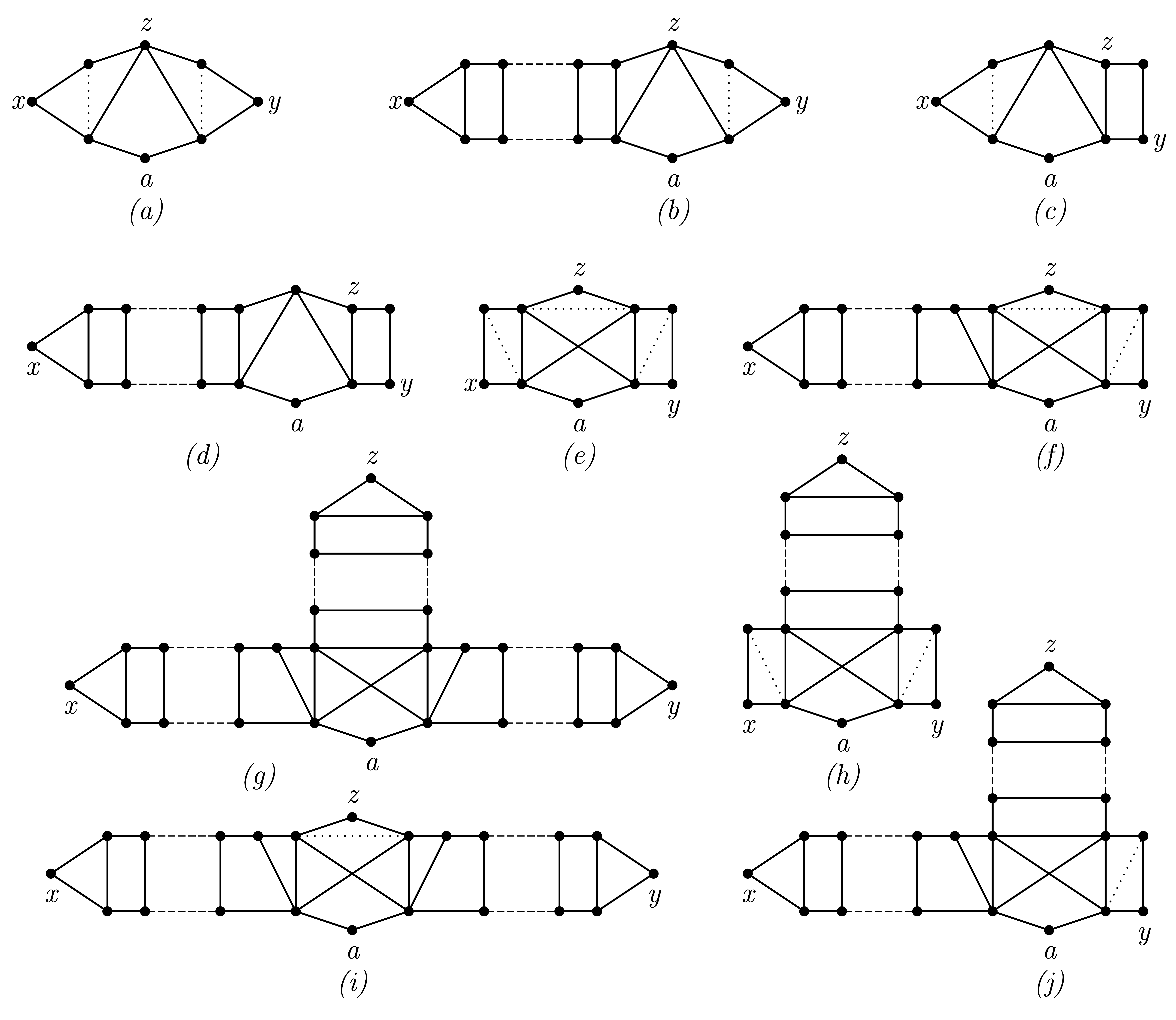}
		\caption{Forbidden 4-chordal subgraphs with slimness 2. The dotted edges may or may not be present. E.g., (a) actually contains three forbidden subgraphs: with both dotted edges present, with only one present, none is present.}
		\label{fig:eightforbidden4Chordal}
	\end{figure}
\end{theorem}

\begin{proof} It is easy to verify that each graph in Figure \ref{fig:eightforbidden4Chordal} has slimness 2. In all  graphs consider $\bigtriangleup(x,y,z)$ and vertex $a\in P(x,y)$.
	
	For the other direction, consider 
	three vertices $x,y,z$ in a 4-chordal graph $G$ and three shortest paths $P(x,y)$, $P(x,z)$ and $P(y,z)$ connecting appropriate pairs.  Assume, by way of contradiction, that there is a vertex $a$ on $P(x,y)$ which is not adjacent to any vertex on $P(x,z)\cup{P(y,z)}$ (i.e., the slimness of $G$ is greater than 1). Assume also that the triple $x,y,z$ for which such a vertex $a$ exists is chosen with the minimum sum $d(x,y)+d(y,z)+d(z,x)$. Let $b$ be a neighbor of $a$ on $P(x,y)$ which is closer to $y$ (see Figure \ref{fig:case1.1ofproproof}).
	
	By the Cycle Lemma applied to $C:=P(x,y)\cup P(x,z)\cup P(y,z)$ and edge $ab$, there must exist vertices $t$ and $c$ in $C$ such that $ac\in E$, $bt\in E$ and $d(t,c)\leq 1$. Since $a$ is not adjacent to any vertex from $P(x,z)\cup{P(y,z)}$, vertex $c$ must belong to $P(x,y)$. As $P(x,y)$ is a shortest path, $t$ cannot belong to $P(x,y)$, implying $d(t,c)= 1$ and $t\in P(x,z)\cup P(y,z)$. See Figure \ref{fig:case1.1ofproproof} for an illustration. By symmetry, in what follows, we will assume that $t$ is on $P(x,z)$. 
	To facilitate our further discussion, we introduce  a few more vertices on paths $P(x,y)$, $P(x,z)$ and $P(y,z)$  (see Figure \ref{fig:case4ofproof}).
	
	\begin{figure}    [htbp]
		\centering
		\includegraphics[width=.24\linewidth]{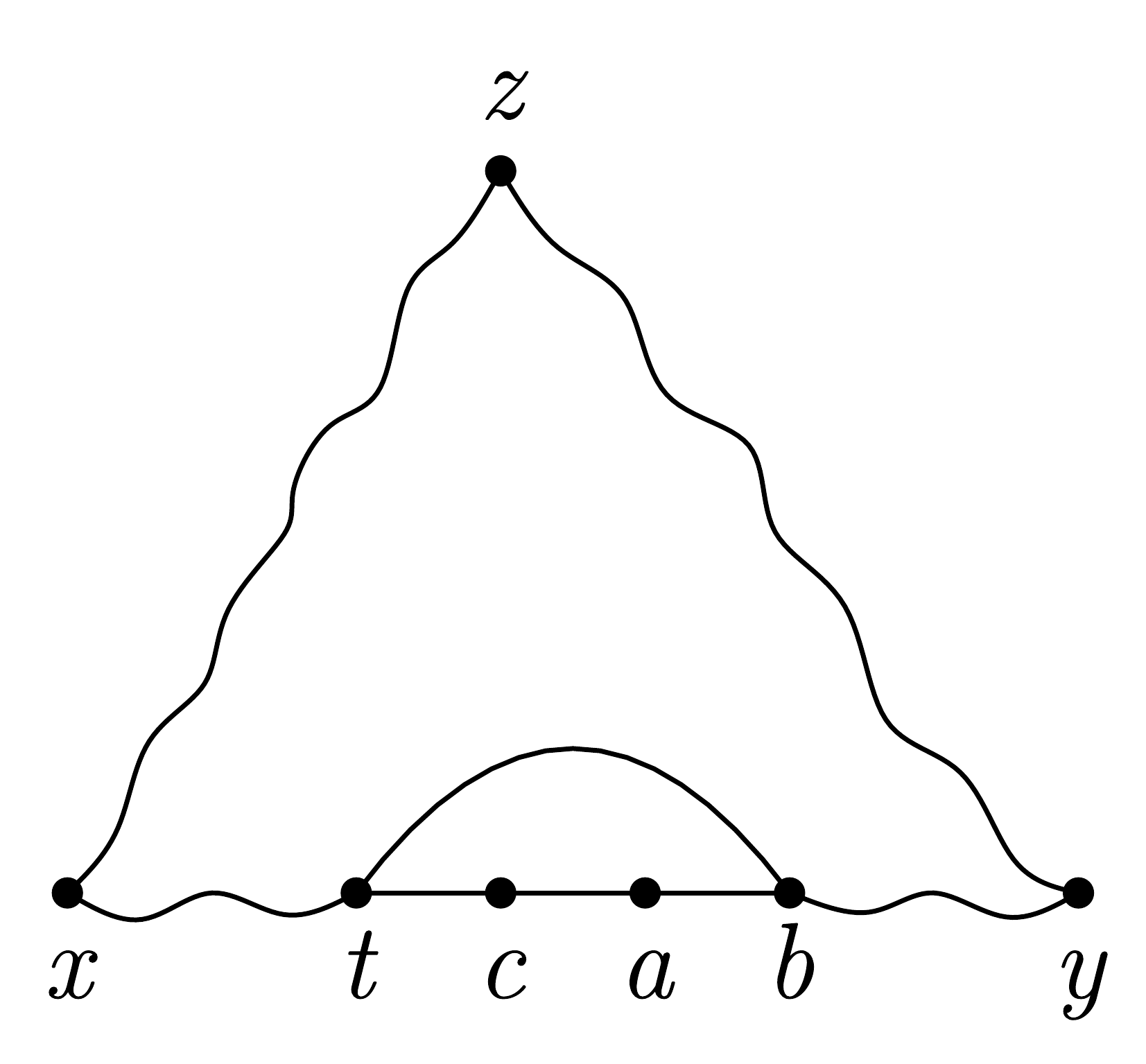}
		\includegraphics[width=.24\linewidth]{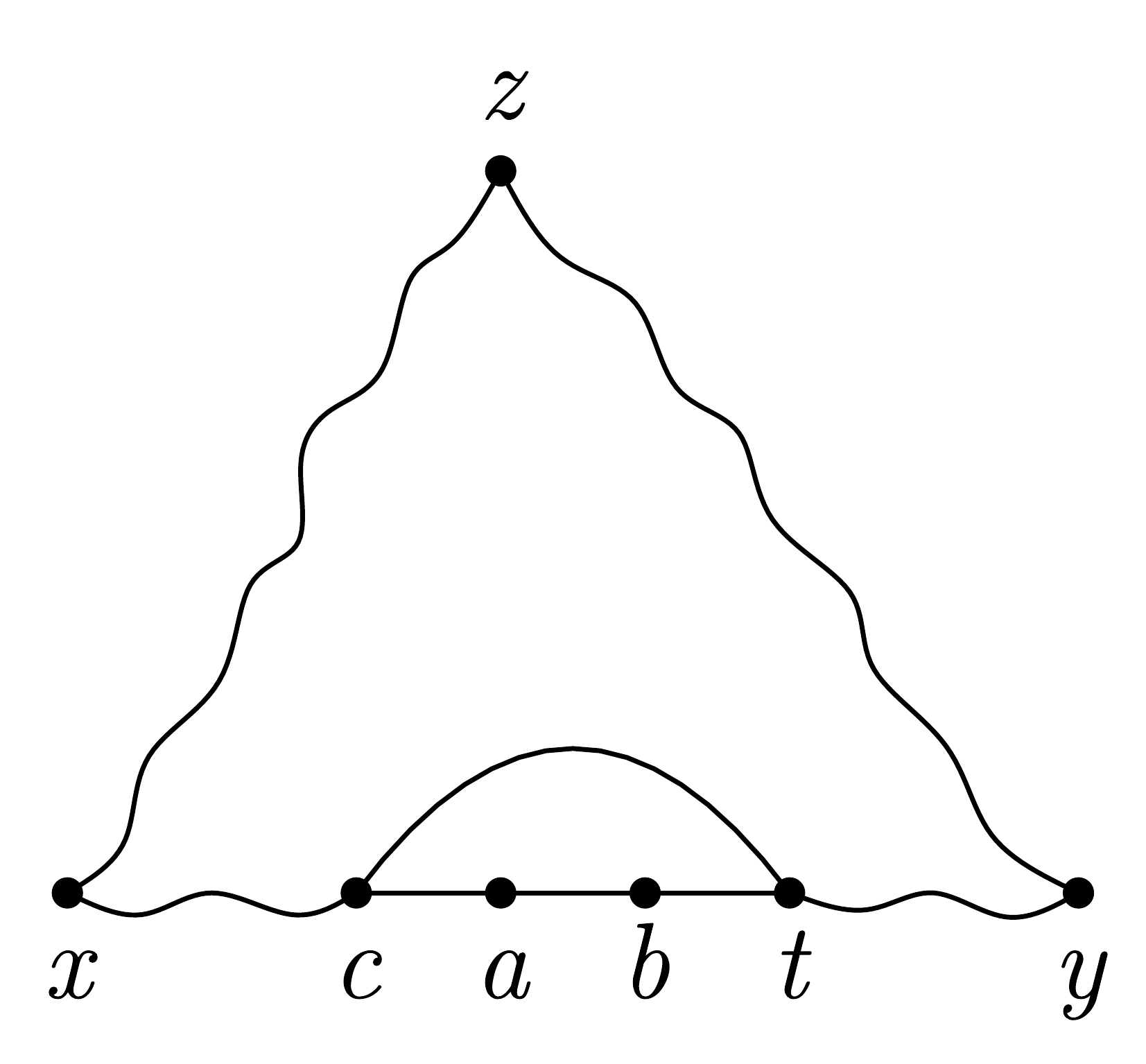}
		\includegraphics[width=.24\linewidth]{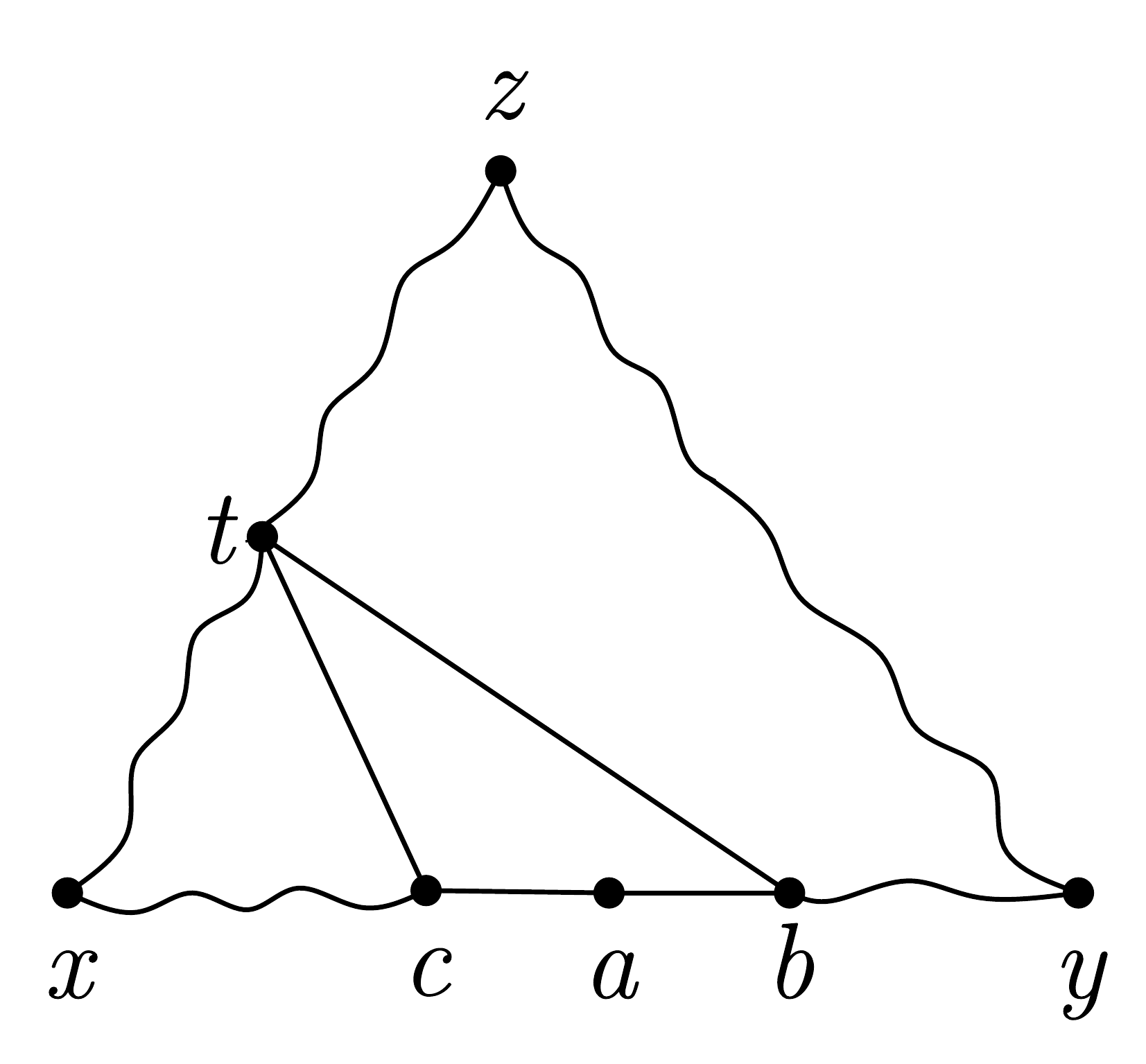}
		\includegraphics[width=.24\linewidth]{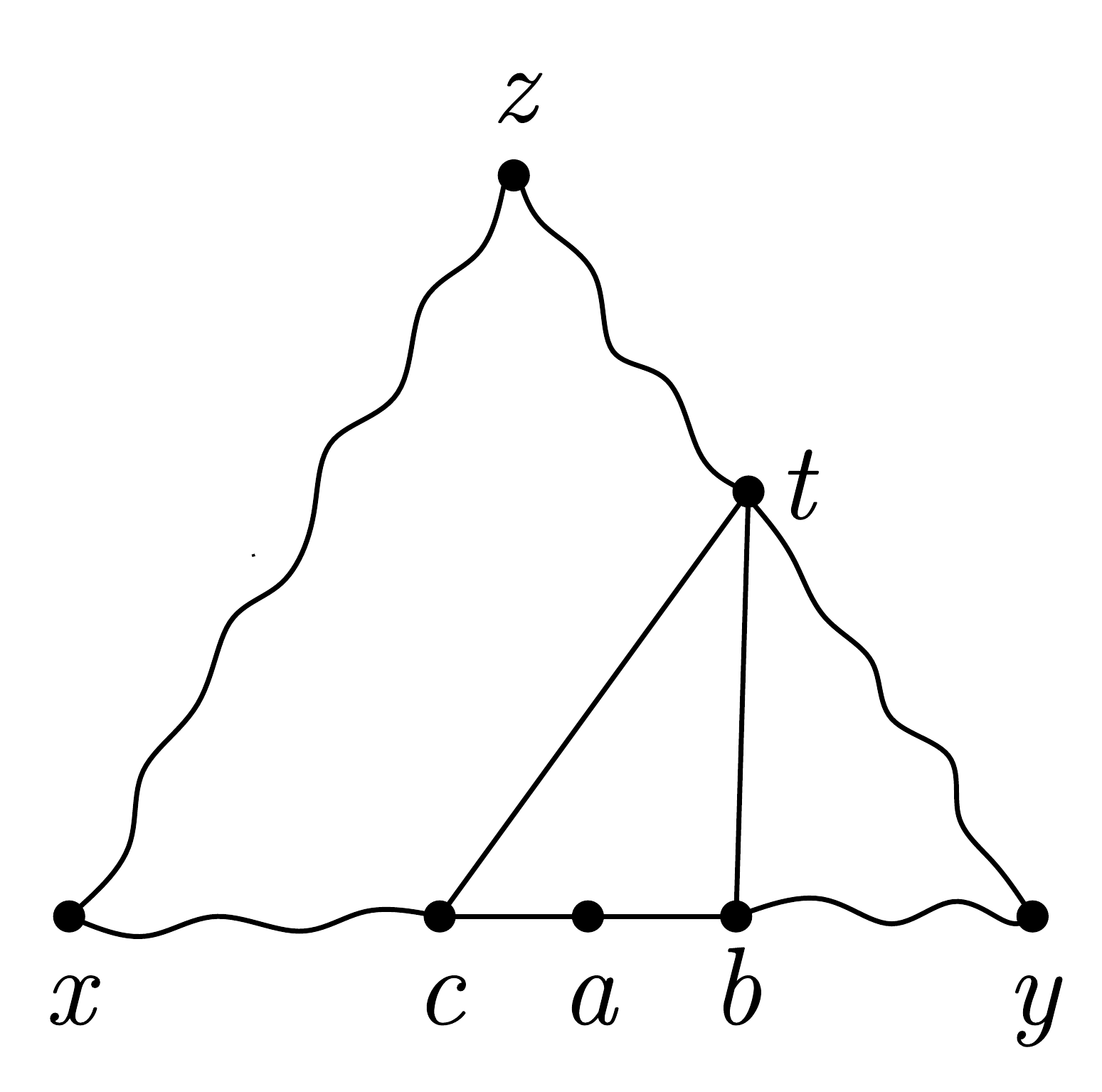}
		\caption{Illustration to the proof  of Theorem \ref{TH:4-chordal_graph}: (left pictures) both $t$ and $c$ are on $P(x,y)$; (right pictures) $c$ is on $P(x,y)$ but $t$ is not.}
		\label{fig:case1.1ofproproof}
		
	\end{figure}
	\begin{figure}    [htbp]
		\centering
		\includegraphics[scale=0.28]{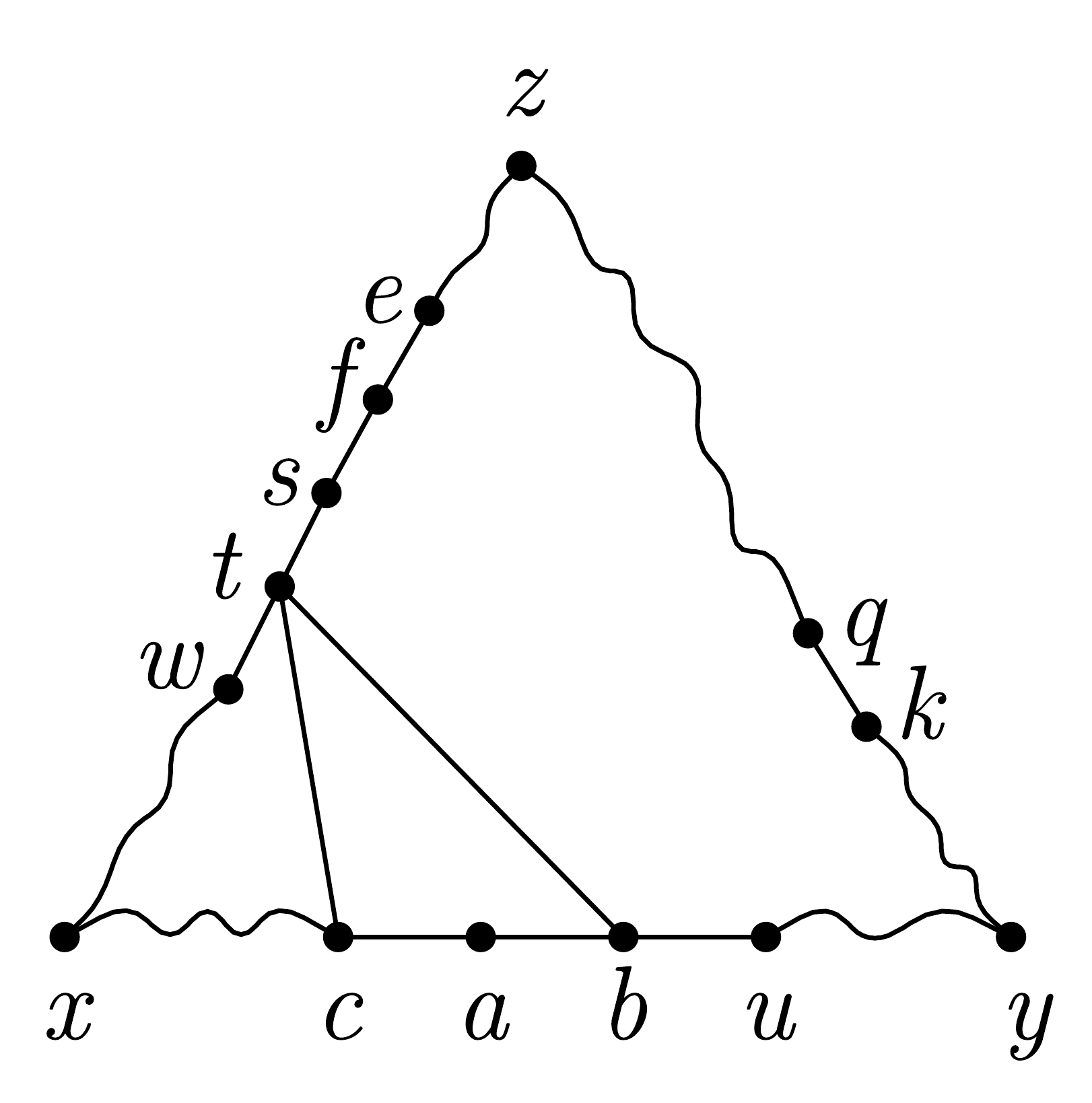}
		\caption{Illustration to the proof  of Theorem \ref{TH:4-chordal_graph}: $t$ is on $P(x,z)$.}
		\label{fig:case4ofproof}
	\end{figure}
	
	Since paths $P(x,z)$ and $P(x,y)$ are shortest, $d(x,t)\leq d(x,c)+1$ and $d(x,t)+1\geq d(x,c)+2$ must hold.
	Therefore,
	\begin{equation}\label{eq:2}
		d(x,t)=d(x,c)+1=d(x,a).
	\end{equation}
	
	From $d(x,t)=d(x,a)$, we also get $d(x,w)=d(x,c)$ (see Figure \ref{fig:case4ofproof}).
	If vertices $w$ and $c$ are adjacent then, by Lemma \ref{lm:multistory_building}, there is
	an isometric building B$(x'|wc)$ in $G$, where $x'\in I(x,w)\cap I(x,c)$. If vertices $w$ and $c$ are not adjacent then, by the Cycle Lemma applied to $C:=P(x,c) \cup P(x,t) \cup {tc}$ and edge $tc$, there must exist a vertex $g\in P(x,c)\cup P(x,t)$ at distance $d(x,c)-1$ from $x$ which forms  a $C_4$ with $w,t,c$ (see Figure \ref{fig:case2ofproproof}). Note that $c\neq x$ as $a$ is not adjacent to any vertex of $P(x,z)$.  This completes the construction of left parts (left to the vertex $a$) in all  forbidden subgraphs from Figure \ref{fig:eightforbidden4Chordal}. If $t=z$ then, by symmetry, we get the same construction on the right side of $a$ and, therefore, get either one of the graphs (a)-(b) from Figure \ref{fig:eightforbidden4Chordal} as an induced subgraph of $G$ or a graph that has a graph (d) from Figure \ref{fig:eightforbidden4Chordal} as an induced subgraph. So, in what follows, we will assume that $t\neq z$.
	
	\begin{figure}    [htbp]
		\centering
		\includegraphics[width=.24\linewidth]{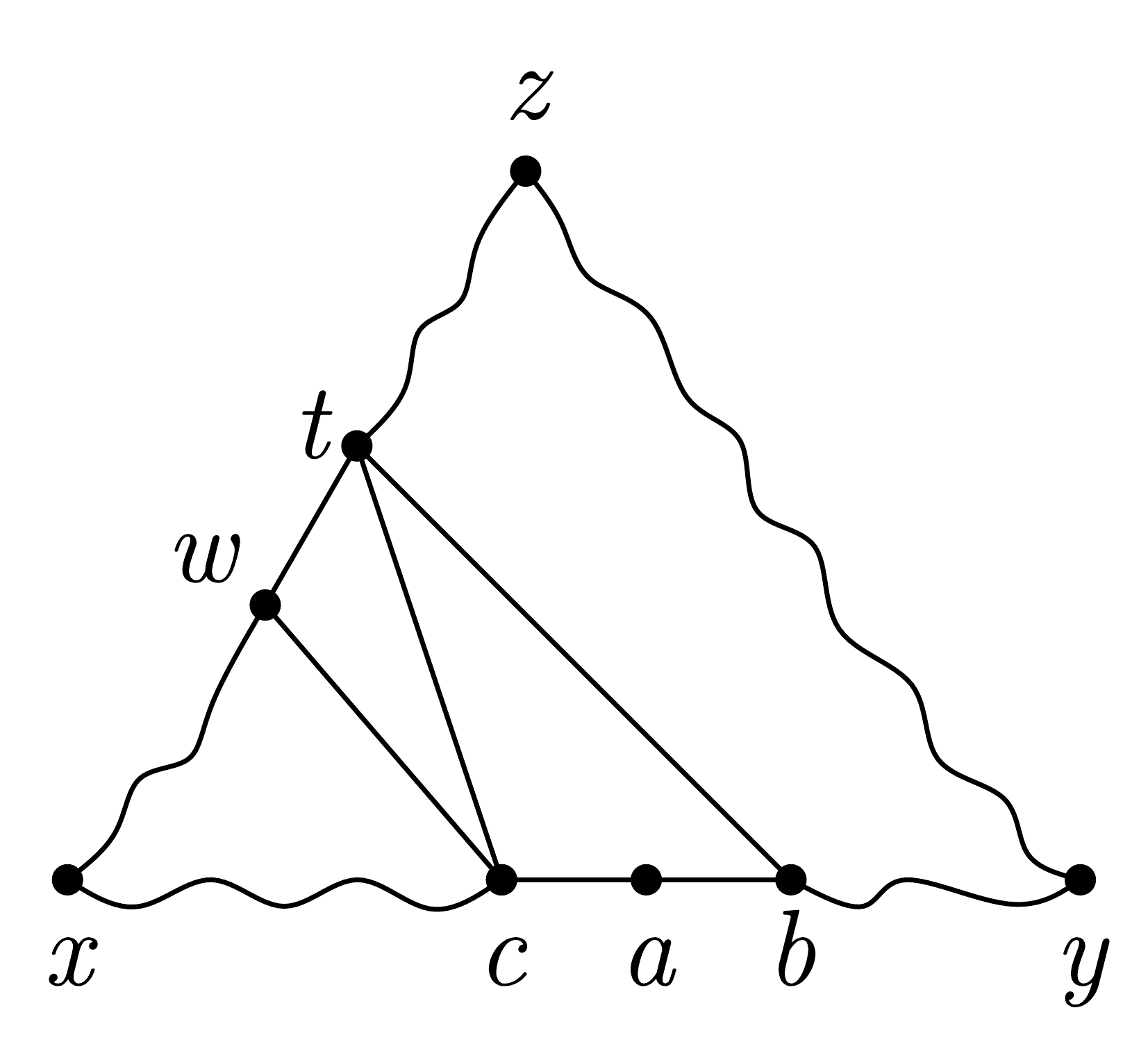}
		\includegraphics[width=.24\linewidth]{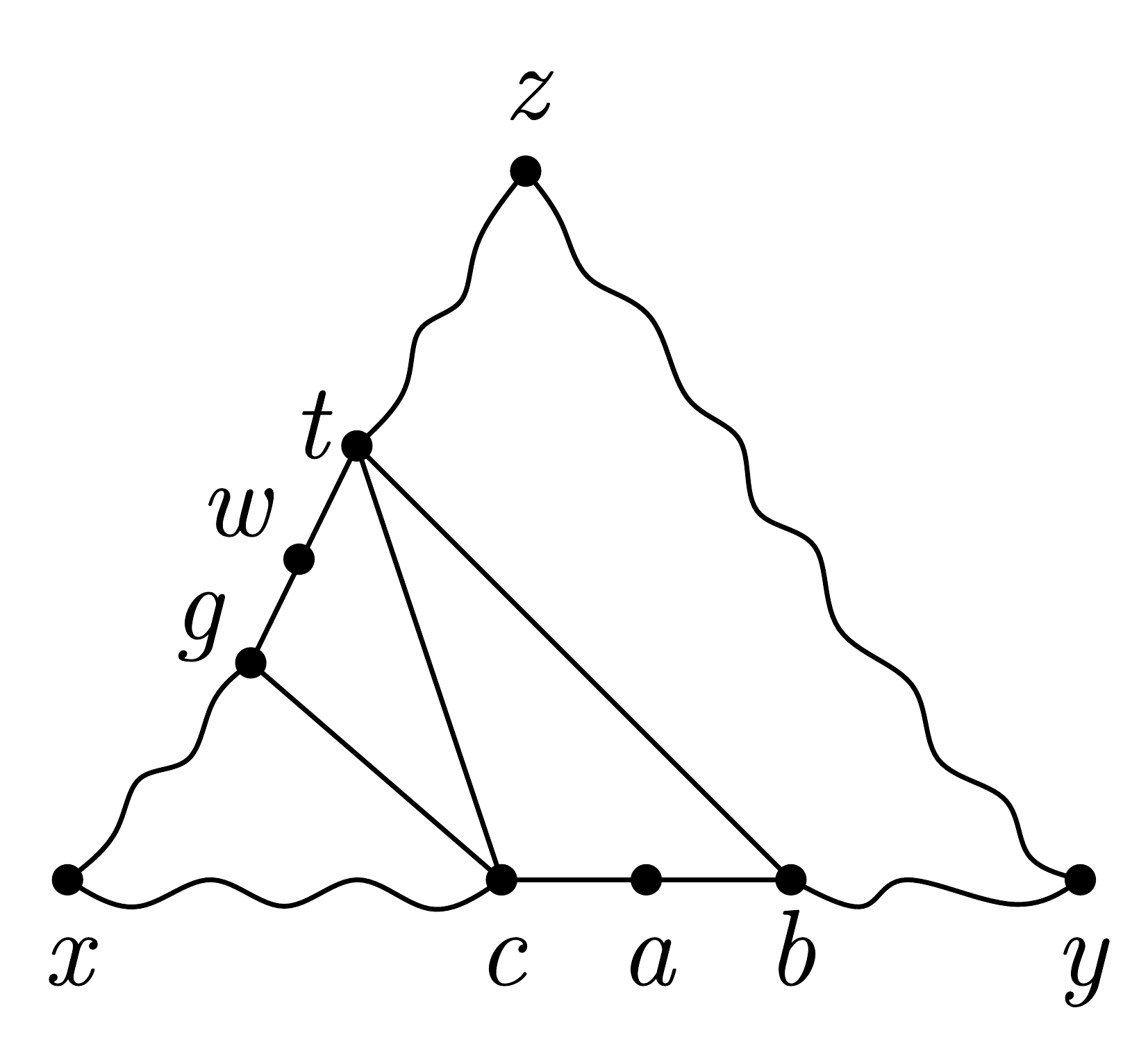}
		\includegraphics[width=.24\linewidth]{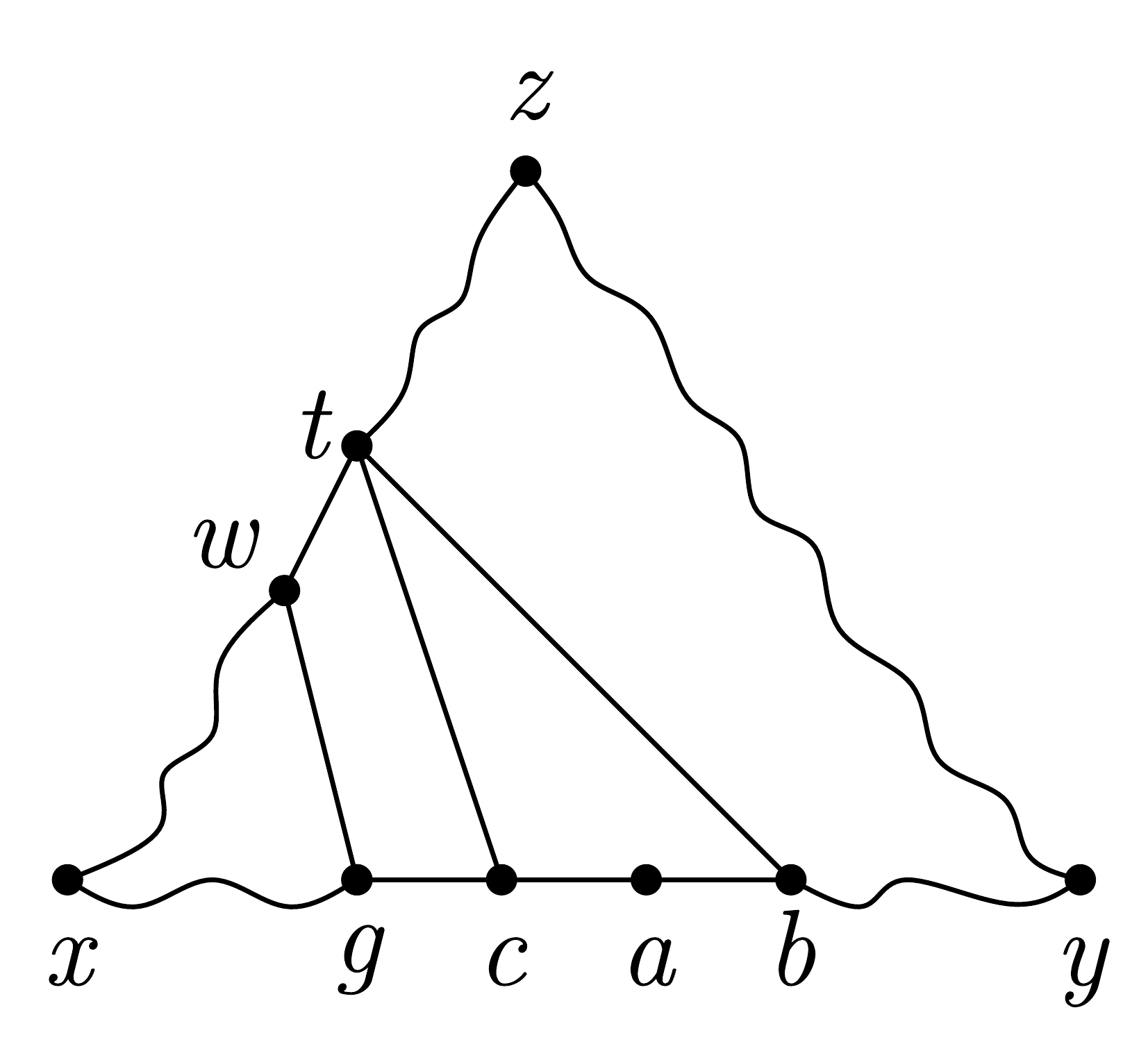}
		\caption{Illustration to the proof  of Theorem \ref{TH:4-chordal_graph}: (left picture)  $c$ and $w$ are adjacent; (right and middle pictures) $c$ and $w$ are not adjacent and there exists a vertex $g\in P(x,c)\cup P(x,t)$ at distance $d(x,c)-1$ from $x$.}
		\label{fig:case2ofproproof}
	\end{figure}
	
	Also, in what follows, if $P(u,v)$ is a shortest path of $G$ and $w$ is a vertex of  $P(u,v)$, then by $P(u,w)$ we denote a subpath of $P(u,v)$ connecting $u$ and $w$.
	
	We will distinguish between two cases, whether vertices $c$ and $b$ have a common neighbor on path $P(y,z)$ or not.
	
	\medskip
	
	\noindent
	\emph{Case 1: Vertices $c$ and $b$ have also a common neighbor $t'$ on path $P(y,z)$.}
	
	\medskip
	\noindent
	First, notice that vertices $w$ and $t'$ cannot be adjacent by the minimality of
	the sum $d(x,y)+d(y,z)+d(z,x)$ (in case $t'w\in E$, consider geodesic triangle  $\bigtriangleup(x,y,t')$, formed by shortest paths $P(x,y),  P(y,t')$ and $P(x,w)\cup\{t'\}$, and vertex $a$ from $P(x,y)$, where $P(y,t')$ and $P(x,w)$ are subpath of shortest paths $P(y,z)$ and $P(x,y)$, respectively). Additionally, by distance requirements, $t'$ is not adjacent to any other vertices from $P(x,w)\cup P(x,c)\setminus \{c\}$ (otherwise, there is a path from $x$ to $b$ via $t'$ which is shorter than $P(x,b)$).
	
	With similar to earlier arguments, we get the needed  construction on the right side of $a$ in all  forbidden subgraphs (e)-(j) from Figure \ref{fig:eightforbidden4Chordal}. Furthermore, as $P(x,z), P(y,z)$ are shortest paths and $d(x,t)=d(x,c)+1$, $d(y,b)+1=d(y,t')$, necessarily, $d(z,t)\le d(z,t')$ and $d(z,t')\le d(z,t)$, i.e., $d(z,t')=(z,t)$.
	If vertices $t$ and $t'$ are adjacent then, by Lemma \ref{lm:multistory_building}, there is an isometric   building B$(z'|tt')$ in $G$, where $z'\in I(z,t)\cap I(z,t')$. If vertices $t$ and $t'$ are not adjacent (see Figure \ref{fig:case3ofproproof}) then, by the Cycle Lemma applied to $C:=P(z,t) \cup P(z,t') \cup \{ct,ct'\}$ and edge $tc$, there must exist a vertex $g\in P(z,t)\cup P(z,t')$ at distance $d(z,t)-1$ from $z$ which forms  a $C_4$ with $t,c,t'$.
	This completes the construction of upper parts (upper to the vertex $a$) in all  forbidden subgraphs (e)-(j) from Figure \ref{fig:eightforbidden4Chordal}.
	
	\begin{figure}[htbp]
		\centering
		\includegraphics[width=.24\linewidth]{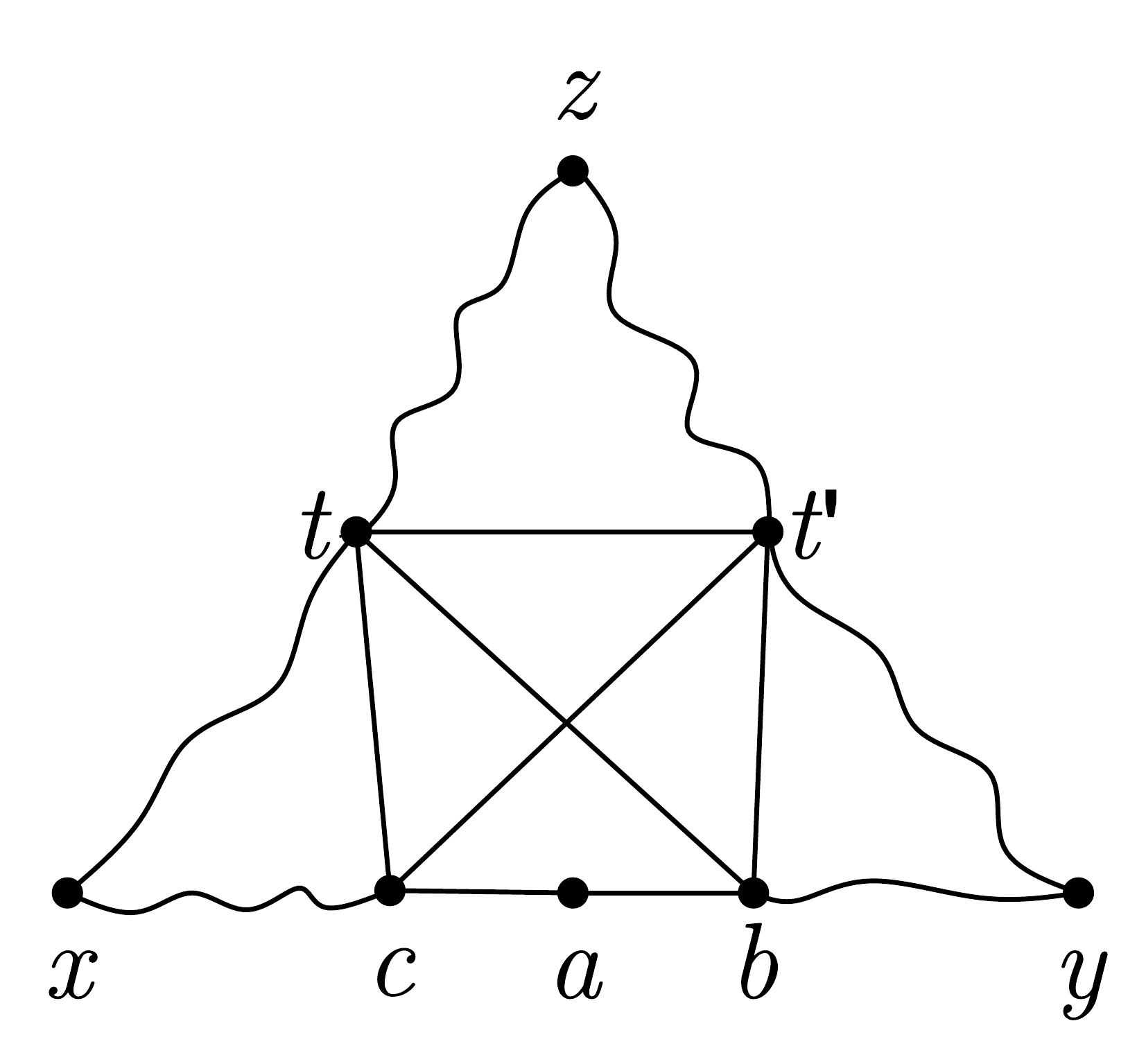}
		\includegraphics[width=.24\linewidth]{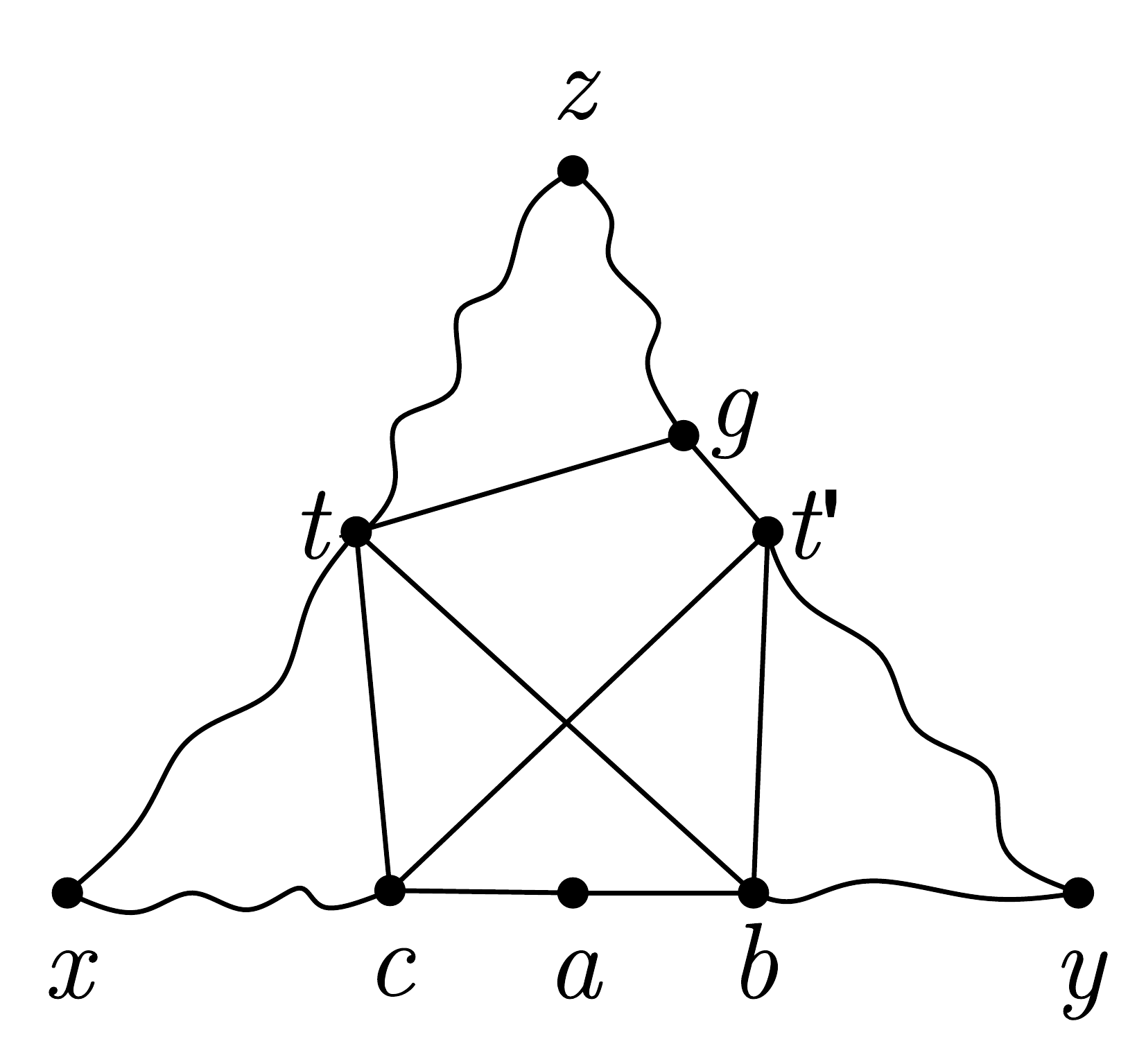}
		\includegraphics[width=.24\linewidth]{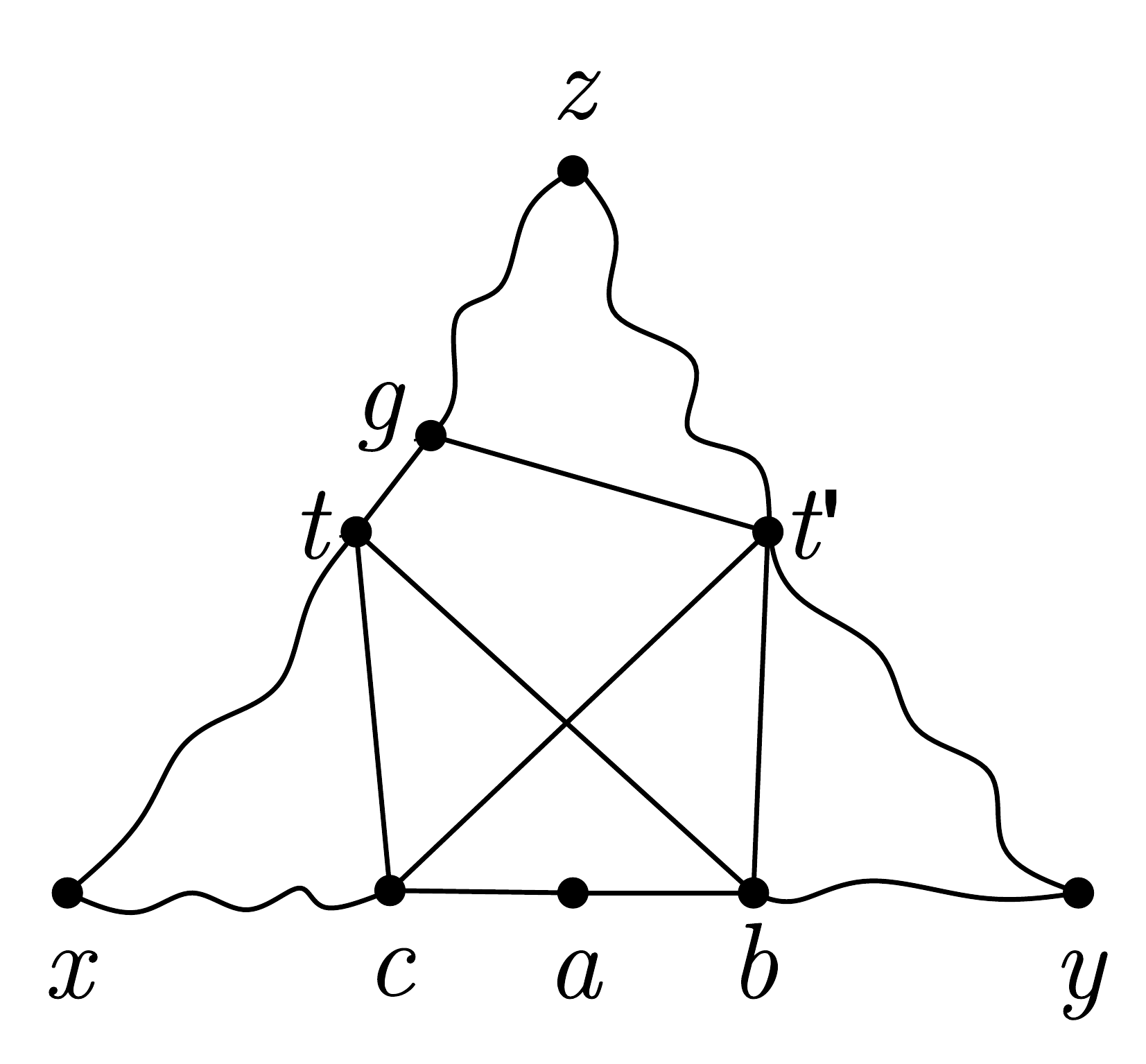}
		\caption{Illustration to the proof  of Theorem \ref{TH:4-chordal_graph}, Case 1:  (left picture) $t$ and $t'$ are adjacent; (right and middle pictures)  $t$ and $t'$ are not adjacent and there exists a vertex $g\in P(z,t)\cup P(z,t')$ at distance $d(z,t)-1$ from $z$.}
		\label{fig:case3ofproproof}
	\end{figure}
	
	\medskip
	
	\noindent
	\emph{Case 2: Vertices $c$ and $b$ have no common neighbors on path $P(y,z)$.}
	
	\medskip
	\noindent
	In what follows, we can assume that $b$ is not adjacent to any vertex of path $P(f,z)\subset P(t,z)$ (see Figure \ref{fig:case4ofproof}).
	If $b$ is adjacent to $f$, then vertices $f,s,t,b$ induce a $C_4$, when $bs\notin E$, or a diamond, otherwise. Hence, $G$ contains
	one of the graphs (a)-(b) from Figure \ref{fig:eightforbidden4Chordal} as an induced subgraph. Furthermore, since $P(t,z)$ is a shortest path, $b$ cannot have neighbors in $P(e,z)\subset P(t,z)$ (otherwise, there is a path from $t$ to $z$ via $b$ which is shorter than $P(t,z)$).

	Consider now the neighbor $u$ of $b$ on path $P(b,y)\subset P(x,y)$ (see Figure \ref{fig:case4ofproof}). We distinguish between three subcases.
	
	\medskip
	
	\noindent
	\emph{Case 2.1: Vertex $u$ is adjacent to a vertex of path $P(t,z)$.}
	
	\medskip
	\noindent
	Since $P(x,y)$ and $P(x,z)$ are shortest paths, $u$ is not adjacent to $t$ (otherwise, $d(c,u)=2<3$) and can be adjacent only to $s$, $f$ or $e$ (otherwise, there is a path from $t$ to $z$ via $u$ which is shorter than $P(t,z)$).
	If $u$ is adjacent to $s$, then vertices $u,s,t,b$ induce a $C_4$, when $bs\notin E$, or a diamond, otherwise. Hence, $G$ contains
	one of the graphs (a)-(b) from Figure \ref{fig:eightforbidden4Chordal} as an induced subgraph.
	If $u$ is adjacent to $f$ but not to $s$, then vertices $u,f,s,t,b$ form a 5-cycle. This cycle must have a chord. The only possible chord is $bs$ (see Figure \ref{fig:added}).  Hence, $G$ contains one of the graphs (c)-(d) from Figure \ref{fig:eightforbidden4Chordal} as an induced subgraph.
	If $u$ is adjacent to $e$ but not to $s,f$, then vertices $u,e,f,s,t,b$ form a 6-cycle. This cycle must have chords and the only possible chord is $bs$. Hence, we get a forbidden $C_5$ or $C_6$ in $G$.  We conclude that $u$ has no neighbors in $P(t,z)\subset P(x,z)$.
	
	\begin{figure}[htbp]
		\centering
		\includegraphics[width=.24\linewidth]{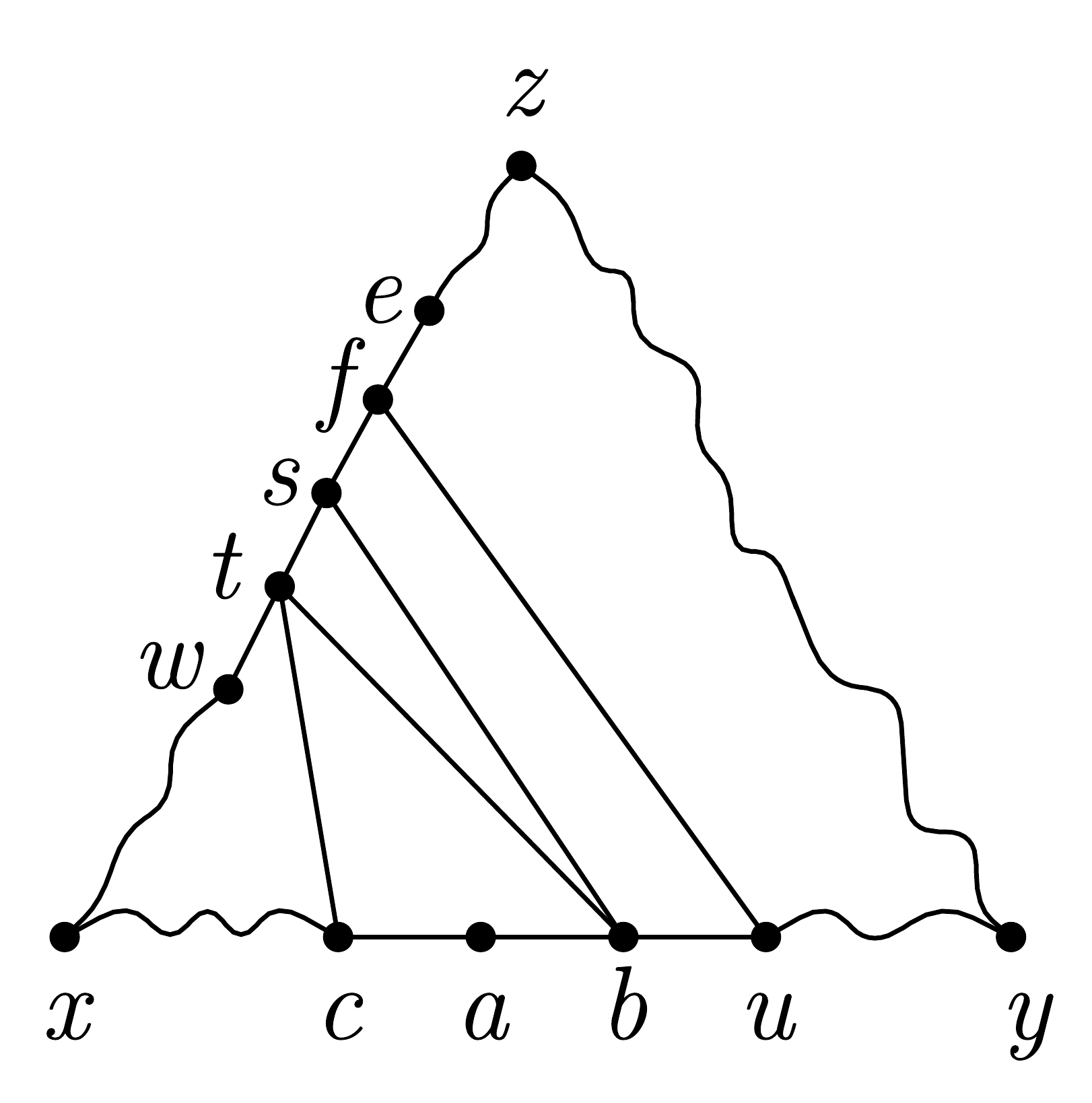}
		\includegraphics[width=.24\linewidth]{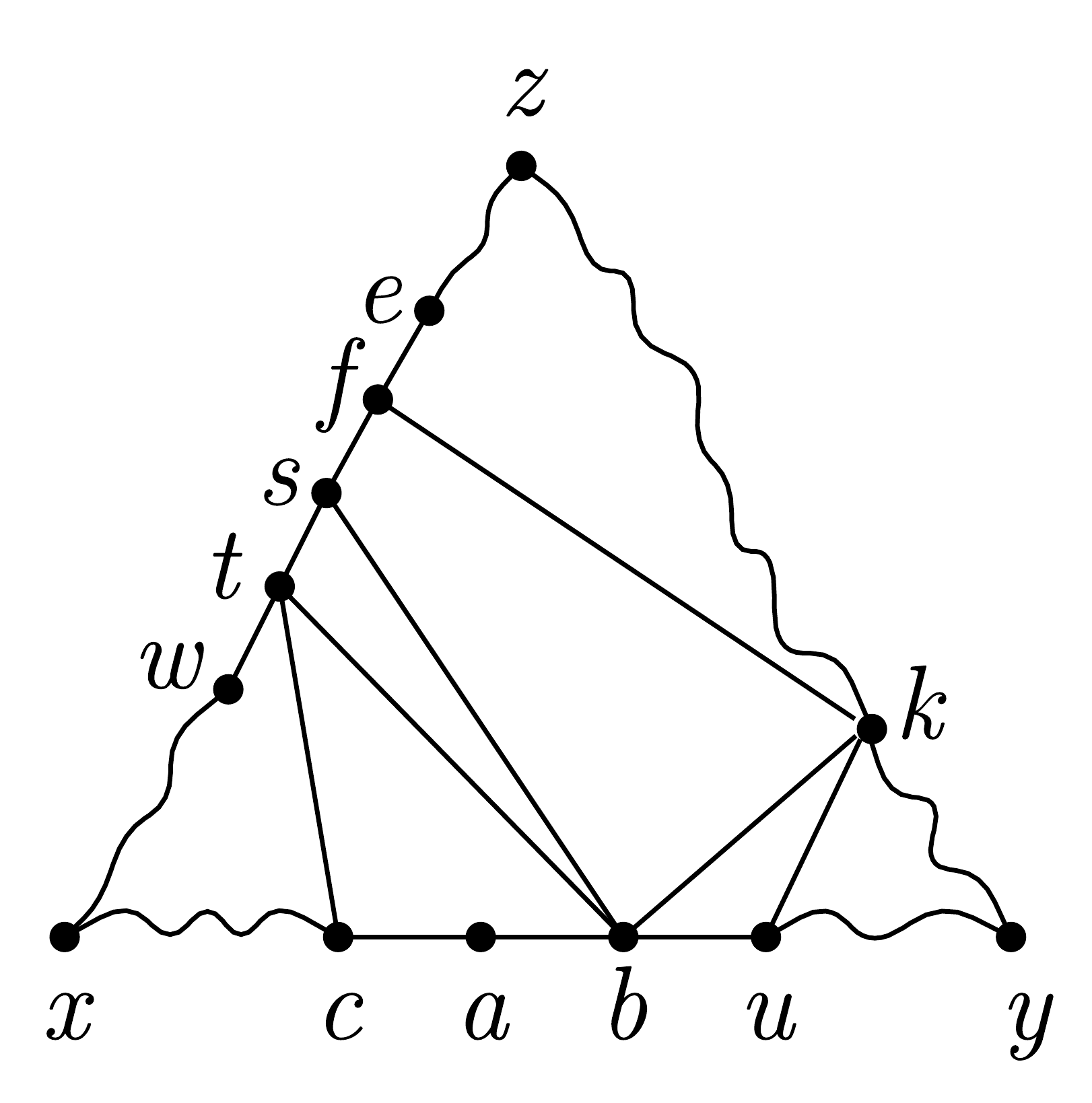}
		\includegraphics[width=.24\linewidth]{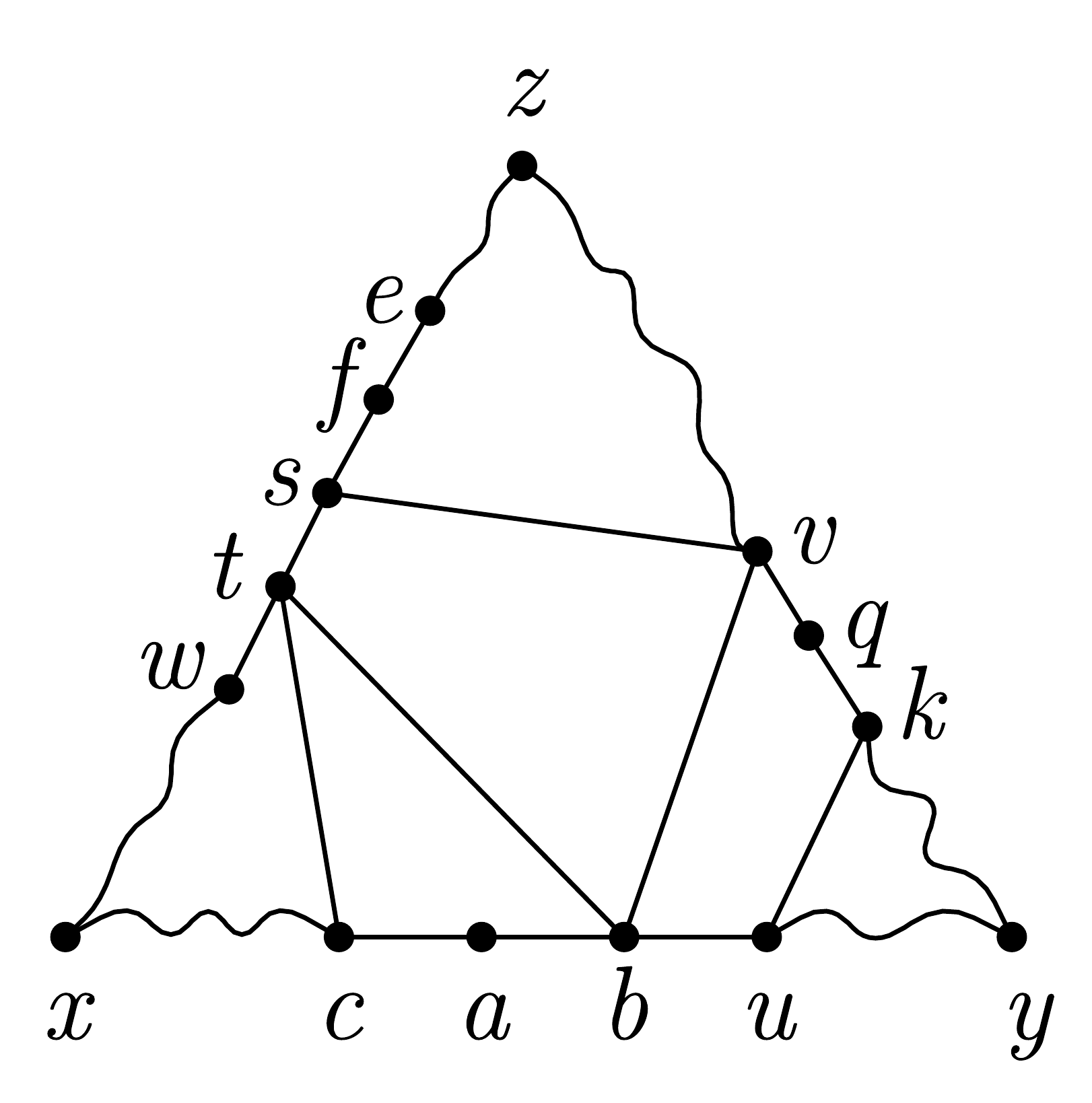}
		\caption{Illustration to the proof  of Theorem \ref{TH:4-chordal_graph}, Case 2: (left picture) $u$ is adjacent to $f$; (middle picture)  $k$ is adjacent to $f$; (right picture) $v$ is adjacent to $s$.}
		\label{fig:added}
	\end{figure}
	
	\medskip
	
	\noindent
	\emph{Case 2.2: Vertex $u$ has no neighbors in $P(t,z)$ but has a neighbor in  $P(y,z)$.}
	
	\medskip
	\noindent
	Let $k$ be a neighbor of $u$ in $P(y,z)$ closest to $z$.
	Hence, $u$ has no other neighbors in $P(k,z)\subset P(y,z)$ except $k$. By the distance requirements, $k$ is not adjacent to any vertex from $P(x,w)\cup  P(x,c)$  (otherwise, there is a path from $x$ to $u$ via $k$ which is shorter than $P(x,u)$).
	Assume $k$ has a neighbor in $P(t,z)$ and pick such a neighbor $v$ which is closest to $t$. Consider a cycle $Z$ formed by vertices $b,u,k$ and subpath $P(t,v)$ of $P(t,z)$. In this cycle of length at least 4, the only possible chords are $bs$ and $bk$. To avoid induced cycles of length at least 5, $v$ must coincide with $t,s$ or $f$.
	If $v=t$, then vertices $u,k,t,b$ induce a $C_4$, when $bk\notin E$, or a diamond, otherwise. Hence, $G$ contains
	one of the graphs (a)-(b) from Figure \ref{fig:eightforbidden4Chordal} as an induced subgraph.
	If $v=s$, then vertices $u,k,s,t,b$ form a 5-cycle. This cycle must have chords. The only possible chords are $bk$ and $bs$.  Hence, $G$ contains one of the graphs (a)-(d) from Figure \ref{fig:eightforbidden4Chordal} as an induced subgraph.
	If $v=f$, then vertices $u,k,f,s,t,b$ form a 6-cycle (see Figure \ref{fig:added}). This cycle must have chords and the only possible chords are $bs$ and $bk$. Hence, again $G$ contains one of the graphs (c)-(d) from Figure \ref{fig:eightforbidden4Chordal} as an induced subgraph.
	
	Thus, $k$ cannot have any neighbors in $P(t,z)$. Recall also that $u$ has no neighbors in $P(t,z)\cup P(k,z)\setminus\{k\}$. Applying the Cycle Lemma to edge $bu$, we get $bk\in E$ or $bq\in E$, where  $q$ is the neighbor of $k$ in $P(k,z)$.

	Denote by $v$ the neighbor of $b$ in $P(k,z)$ closest to $z$. Note that
	$v$ is not adjacent to $c$ as $c$ and $b$ do not have common neighbors in $P(y,z)$. By the minimality of the sum $d(x,y)+d(y,z)+d(z,x)$, $v$ is not adjacent to $w$ (in case $vw\in E$, consider geodesic triangle $\bigtriangleup(x,y,v)$, formed by shortest paths $P(x,y),  P(y,v)$ and $P(x,w)\cup\{v\}$, and vertex $a$ from $P(x,y)$). Additionally, by distance requirements, $v$ is not adjacent to any other vertices from $P(x,w)\cup P(x,c)$  (otherwise, there is a path from $x$ to $b$ via $v$ which is shorter than $P(x,b)$).
	
	As $P(k,z)$ is a shortest path, $d(v,k)\le 3$, when $bk\notin E$, and $d(v,k)\le 2$, when $bk\in E$.  Applying the Cycle Lemma to edge $bv$, we  get $vt\in E$, or $vs\in E$, or $vf\in E$ (and hence, $bs\in E$, otherwise, we get a $C_5$ in $G$), or $G$ contains an induced $C_4$ formed by vertices $b,s,v',v$ or vertices $t,b,v,v'$, where $v'$ is the neighbor of $v$ in $P(v,z)$.
	If $v$ is adjacent to $t$, then vertices $t,b$ with some vertices of path $P(v,k)\subset P(y,z)$ and possibly with $u$ (e.g., when $v=q$) induce a $C_4$ or a house or a diamond in $G$ (we leave these small technical details to the reader; one needs to analyze a subgraph on at most 7 vertices). Note that in this case $v\neq k$ as $k$ has no neighbors in $P(t,z)$.  So,   $G$ contains one of the graphs (a)-(d) from Figure \ref{fig:eightforbidden4Chordal} as an induced subgraph. Now, we assume that $tv\notin E$.
	If $v$ is adjacent to $s$, then vertices $v,s,t,b$ induce a $C_4$, when $bs\notin E$, or a diamond, otherwise  (see Figure \ref{fig:added}). Hence, $G$ contains
	one of the graphs (a)-(b) from Figure \ref{fig:eightforbidden4Chordal} as an induced subgraph.
	If $v$ is adjacent to $f$ but not to $s$, then vertices $v,f,s,b$ form an induced  $C_4$. Hence, $G$ contains one of the graphs (c)-(d) from Figure \ref{fig:eightforbidden4Chordal} as an induced subgraph. So, we can assume that $tv,vs,vf\notin E$.

	If vertices $b,s,v',v$ or vertices $t,b,v,v'$ form an induced $C_4$ in $G$, then $G$ contains one of the graphs (a)-(d) from Figure \ref{fig:eightforbidden4Chordal} as an induced subgraph. In this case, one needs to notice only that $v'$ is not adjacent to any vertex from $P(x,w)\cup P(x,c)$. If $v'$ is adjacent to a vertex $w'$ from $P(x,w)\cup P(x,c)\setminus\{c,w\}$, then we get a contradiction 
	with distance requirements (subpath $P(x,b)$ of a shortest path $P(x,y)$ is a shortest path) or with the minimality of the sum $d(x,y)+d(y,z)+d(z,x)$ (consider geodesic triangle $\bigtriangleup(x,y,v')$, formed by shortest paths $P(x,y),  P(y,v')$ and $P(x,w')\cup\{v'\}$, and vertex $a$ from $P(x,y)$). If $v'$ is adjacent to $c$ then $G$ has an induced $C_5$ formed by $c,a,b,v,v'$, which is impossible.  If $v'$ is adjacent to $w$ then $G$ has an induced $C_6$ formed by $c,a,b,v,v',w$, when $wc\in E$, or an induced $C_7$ formed by $c,a,b,v,v',w,g$, otherwise (vertex $g$ can be seen in Figure \ref{fig:case2ofproproof}). Both induced cycles are forbidden in $G$.
	
	\medskip
	
	\noindent
	\emph{Case 2.3: Vertex $u$ has no neighbors in $P(y,z)\cup P(t,z)$.}
	
	\medskip
	\noindent
	In this case, we have a contradiction to our assumption that the triple $x,y,z$ certifying $sl(G)>1$ has the smallest sum $d(x,y)+d(y,z)+d(z,x)$. For triple $c,z,y$ with shortest paths $P(c,y)\subset P(x,y)$, $P(z,y)$ and $P(c,z)=\{c\}\cup P(t,z)$, $u\in P(c,y)$ has no neighbors in $P(c,z)\cup P(z,y)$ and $d(x,y)+d(y,z)+d(z,x)> d(c,y)+d(y,z)+d(z,c)$. Here, $P(c,z)=\{c\}\cup P(t,z)$ is a shortest path since $d(x,w)=d(x,c)$ and hence $d(c,z)=d(w,z)=1+d(t,z)$.
\end{proof}

A \emph{ hole} is an induced cycle of length at least 5.  A {\sl domino} is an induced cycle on 6 vertices with one additional chord dividing it
into two cycles of length 4 each. A \emph{ house-hole-domino--free graph} (\emph{HHD-free graph}) is a graph not containing holes, houses and dominoes as induced subgraphs. As 4-chordal  graphs do not contain any induced holes and the graphs shown in Figure \ref{fig:eightforbidden4Chordal} have houses or dominoes as induced subgraphs, we have.

\begin{corollary}
	For every HHD-free graph $G$, $sl(G)\le 1$.
\end{corollary}

\acknowledgements
\label{sec:ack}
We would like to thank anonymous reviewers for many useful suggestions and comments. 

\nocite{*}
\bibliographystyle{abbrvnat}
\bibliography{SlimnessOfGraphs}
\label{sec:biblio}

\end{document}